\documentclass[12pt]{article}

\RequirePackage[OT1]{fontenc}
\usepackage{amsthm,amsmath,natbib}
\RequirePackage[colorlinks,citecolor=blue,urlcolor=blue]{hyperref}

\def\frac#1#2{{\textstyle{#1\over#2}}}

\DeclareSymbolFont{AMSb}{U}{msb}{m}{n}
\DeclareMathSymbol{\Natural}{\mathbin}{AMSb}{"4E}
\DeclareMathSymbol{\Integer}{\mathbin}{AMSb}{"5A}
\DeclareMathSymbol{\Real}{\mathbin}{AMSb}{"52}
\DeclareMathSymbol{\Rational}{\mathbin}{AMSb}{"51}
\DeclareMathSymbol{\Imaginary}{\mathbin}{AMSb}{"49}
\DeclareMathSymbol{\Complex}{\mathbin}{AMSb}{"43} 
\DeclareMathSymbol{\Disk}{\mathbin}{AMSb}{"44} 
\def\bi{\begin{itemize}}
\def\ei{\end{itemize}}
\def\bd{\begin{description}}
\def\ed{\end{description}}
\def\ben{\begin{enumerate}}
\def\een{\end{enumerate}}
\def\bv{\begin{verbatim}}
\def\ev{\end{verbatim}}

\def\calL{{\mathcal L}}

\def\calN{{\mathcal N}}

\def\calO{{\mathcal{O}}}
\def\calS{{\mathcal{S}}}

\def\bar#1{{\overline{#1}}}
\def\hat#1{{\widehat{#1}}}

\newcommand{\indep}{\perp\!\!\!\perp}

\def\pr{{\rm Pr}}

\def\2to{{\ {\buildrel 2\over \longrightarrow}\ }}

\def\Deq{{\ {\buildrel D\over =}\ }}

\def\I1ton{{$I_1,\ldots,I_n$}}
\def\X1ton{{$X_1,\ldots,X_n$}}
\def\Y1ton{{$Y_1,\ldots,Y_n$}}
\def\Z1ton{{$Z_1,\ldots,Z_n$}}
\def\R1ton{{$R_1,\ldots,R_n$}}
\def\e1ton{{$e_1,\ldots,e_n$}}
\def\t1ton{{$t_1,\ldots,t_n$}}
\def\x1ton{{$x_1,\ldots,x_n$}}
\def\y1ton{{$y_1,\ldots,y_n$}}
\def\z1ton{{$z_1,\ldots,z_n$}}

\def\calS{{\mathcal{S}}}

\newtheorem{defn}{Definition}

\newtheorem{theorem}[defn]{Theorem}
\newtheorem{lemma}[defn]{Lemma}

\usepackage[english]{babel}
\usepackage{amsfonts}
\usepackage{amssymb}
\usepackage{graphicx}
\usepackage{enumerate}

\usepackage{ulem}

\def\RemTO#1{}

\def\RemET#1{}

\begin{document}
\thispagestyle{empty}
\baselineskip=28pt
\vskip 5mm
\begin{center} {\Large{\bf Bridging Asymptotic Independence and Dependence in Spatial Extremes Using Gaussian Scale Mixtures}}
\end{center}

\baselineskip=12pt


\vskip 5mm

\begin{center}
\large
Rapha\"el Huser$^1$, Thomas Opitz$^2$ and Emeric Thibaud$^3$
\end{center}

\footnotetext[1]{
\baselineskip=10pt Computer, Electrical and Mathematical Sciences and Engineering (CEMSE) Division, King Abdullah University of Science and Technology (KAUST), Thuwal 23955-6900, Saudi Arabia. E-mail: raphael.huser@kaust.edu.sa}
\footnotetext[2]{
\baselineskip=10pt  INRA, UR546 Biostatistics and Spatial Processes, 228, Route de l'A\'erodrome, CS 40509, 84914 Avignon, France. E-mail: thomas.opitz@inra.fr}
\footnotetext[3]{
\baselineskip=10pt Ecole Polytechnique F\'ed\'erale de Lausanne, EPFL-FSB-MATHAA-STAT, Station 8, 1015 Lausanne, Switzerland. E-mail: emeric.thibaud@epfl.ch}

\baselineskip=17pt
\vskip 4mm
\centerline{\today}
\vskip 6mm

\begin{center}
{\large{\bf Abstract}}
\end{center}

Gaussian scale mixtures are constructed as Gaussian processes with a random variance. They have non-Gaussian marginals and can exhibit asymptotic dependence unlike Gaussian processes, which are asymptotically independent except in the case of perfect dependence. In this paper, we study in detail the extremal dependence properties of Gaussian scale mixtures and we unify and extend general results on their joint tail decay rates in both asymptotic dependence and independence cases. Motivated by the analysis of spatial extremes, we propose several flexible yet parsimonious parametric copula models that smoothly interpolate from asymptotic dependence to independence and include the Gaussian dependence as a special case. We show how these new models can be fitted to high threshold exceedances using a censored likelihood approach, and we demonstrate that they provide valuable information about tail characteristics. Our parametric approach outperforms the widely used nonparametric $\chi$ and $\bar\chi$ statistics often used to guide model choice at an exploratory stage by borrowing strength across locations for better estimation of the asymptotic dependence class. We demonstrate the capacity of our methodology by adequately capturing the extremal properties of wind speed data collected in the Pacific Northwest, US.
\baselineskip=16pt

\par\vfill\noindent
{\bf Keywords:} asymptotic dependence and independence; censored likelihood inference; spatial copula; extreme event; random scale model; threshold exceedance.\\

\pagenumbering{arabic}
\baselineskip=22pt


\newpage

\section{Introduction}

\subsection{Motivation}

Gaussian processes have been used extensively in classical spatial statistics  
thanks to their appealing theoretical properties, tractability in high dimensions, explicit conditional distributions and ease of simulation. However, as far as the modeling of extremes is concerned, they have been heavily criticized \citep{Davison.etal:2013} as being unable to capture asymptotic dependence; Gaussian processes are asymptotically independent
, meaning that the dependence strength between events observed at two distinct spatial locations vanishes  as their extremeness increases. Without firm knowledge about the tail properties of the data, it is safer (i.e., more conservative) in terms of risk of joint extremes to assume asymptotic dependence. Using stochastic processes that lack flexibility in the joint tail may lead to severe under- or overestimation of probabilities associated to simultaneous extreme events. This lack of solid theoretical foundations for extrapolation beyond the range of the observations has been a catalyst for extensive research in extreme-value theory (EVT).

Classical EVT provides support for the use of max-stable models for block-maxima (e.g., annual maxima of daily temperature or precipitation), because they are the only possible limits of renormalized pointwise maxima of spatial processes \citep{deHaan.Ferreira:2006}. Their strong asymptotic justification is both a blessing and a curse: max-stability provides a robust modeling framework when few extreme data are available, but this strong assumption may be far from satisfied at subasymptotic levels arising with finite samples. An instructive example is asymptotic independence, where the limiting max-stable distribution is the product of independent margins and cannot capture the potentially strong dependence that remains at extreme subasymptotic levels. In addition to this possibly large gap between the theory and the data, inference for max-stable models is tricky. Full likelihoods can only be calculated in small dimensions, which led to the use of less efficient inference techniques, such as composite likelihoods \citep{Padoan.etal:2010,Huser.Davison:2013a,Castruccio.etal:2016}. \citet{Thibaud.etal:2015} and \citet{Dombry.etal:2016} recently showed how the full likelihood may be approximated in Bayesian or frequentist settings  by integrating out a random partition using Monte Carlo techniques; however, these approaches remain computer intensive in large dimensions.  Alternatively, models for threshold exceedances based on the limiting Poisson \citep{Wadsworth.Tawn:2014,Engelke.etal:2015} or the generalized Pareto \citep{Ferreira.deHaan:2014,Thibaud.Opitz:2015} process are the counterparts of max-stable models for threshold exceedances and have become increasingly popular because they circumvent many of the computational bottlenecks of max-stable processes. However, analogous to max-stable processes, Pareto processes are threshold-stable and thus lack tail flexibility, especially when fitted to asymptotically independent data. Very extreme joint risks tend to be strongly overestimated by these models if the data exhibit decreasing dependence strength at more extreme levels.

Because of the practical limitations of ultimate models, such as max-stable or Pareto processes, it is natural to seek penultimate (i.e., subasymptotic) models for spatial extremes, which combine tail flexibility with computational tractability and have known tail characteristics, in the same vein as penultimate approximations in univariate EVT \citep{Kaufmann.2000}. 
In the case of asymptotic independence, Gaussian models might be reasonable. Alternatively, \citet{Wadsworth.Tawn:2012} proposed inverted max-stable models, which were found to be slightly more flexible than Gaussian processes in some applications \citep{Thibaud.etal:2013,Davison.etal:2013}, but are as difficult to fit as max-stable models. A more complex Bayesian nonparametric copula model was proposed by \citet{Fuentes.etal:2013}. Recently, \citet{op2015} advocated a very specific Gaussian scale mixture model designed for asymptotic independence, constructed from the product of a standard Gaussian process with a random variance following the exponential distribution, yielding Laplace random fields. 
In the case of asymptotic dependence, subasymptotic models were also developed. 
\citet{Wadsworth.Tawn:2012} 
proposed max-mixtures involving inverted max-stable and max-stable models, which add flexibility over the latter at the price of a relatively large number of parameters to be estimated. 
\citet{Morris.etal:2017} proposed a Bayesian space-time skew-$t$ model for threshold exceedances. \citet{Krupskii.etal:2017} proposed factor copula models constructed from Gaussian processes with a random mean, which can capture asymptotic dependence or independence. 

The above spatial models are useful in many respects by introducing more flexibility or by improving computation over max-stable models, but they focus on modeling exclusively either asymptotically independent or asymptotically dependent data. 
In contrast, the pseudo-polar representations of multivariate limit distributions have motivated  \citet{Wadsworth.al.2017} to explore  how more flexible transitions between dependence classes can be achieved in the bivariate case through a common random scaling  applied to a random vector on the unit sphere, the latter being defined from a norm on $\mathbb{R}^2$. Using the pseudo-polar representation of multivariate elliptical distributions such as the multivariate Gaussian, we argue in this paper that a flexible and natural extension of this approach to spatial modeling consists in using the wide class of randomly scaled Gaussian processes, also known as Gaussian scale mixtures, which comprise all infinite-domain processes with finite-dimensional nondegenerate elliptical distributions  \citep{Huang.Cambanis.1979}. The gain in tail flexibility as compared to Gaussian or asymptotic models may also allow fixing lower thresholds   in exceedance-based modeling. Then, the Gaussian correlation structure may capture certain properties of the bulk of the distribution like the range of dependence, while the random scale parameters give separate control over the  joint tail decay rates.


A main theoretical contribution of this paper is that we give detailed results on the joint tail decay rates of Gaussian scale mixture processes under general assumptions on their random scale, defining conditions to capture asymptotic dependence or asymptotic independence. This tail characterization then leads to our main methodological novelty: we propose new spatial subasymptotic copula models, which smoothly bridge the two asymptotic dependence regimes and allow estimating the latter from the data. The model type is usually chosen a priori using variants of the coefficients $\chi$ and $\bar\chi$ \citep{Coles.etal:1999}, whose nonparametric estimation entails large uncertainties and does not yield a spatially coherent model. To illustrate this, Figure~\ref{fig:results2} displays the nonparametric and model-based estimates of the quantities $\chi(u)$ and $\bar\chi(u)$ defined in \eqref{eq:chi_chibar} for an asymptotically independent process. The parametric estimators of $\chi(u)$ and $\bar\chi(u)$ are much more reliable than their nonparametric counterparts, and our approach allows borrowing strength across locations for better estimation and discrimination between the two asymptotic classes. More details are given in~\S\ref{sec:SimulStudy}. We also demonstrate how to exploit the underlying Gaussian structure of our new models to make inference based on a full likelihood with partial censoring and to efficiently perform conditional simulation, which is much more tricky when max-stable models are involved \citep{Dombry.etal:2013}. 

\begin{figure}[t!]
\centering
\includegraphics[width=0.9\linewidth]{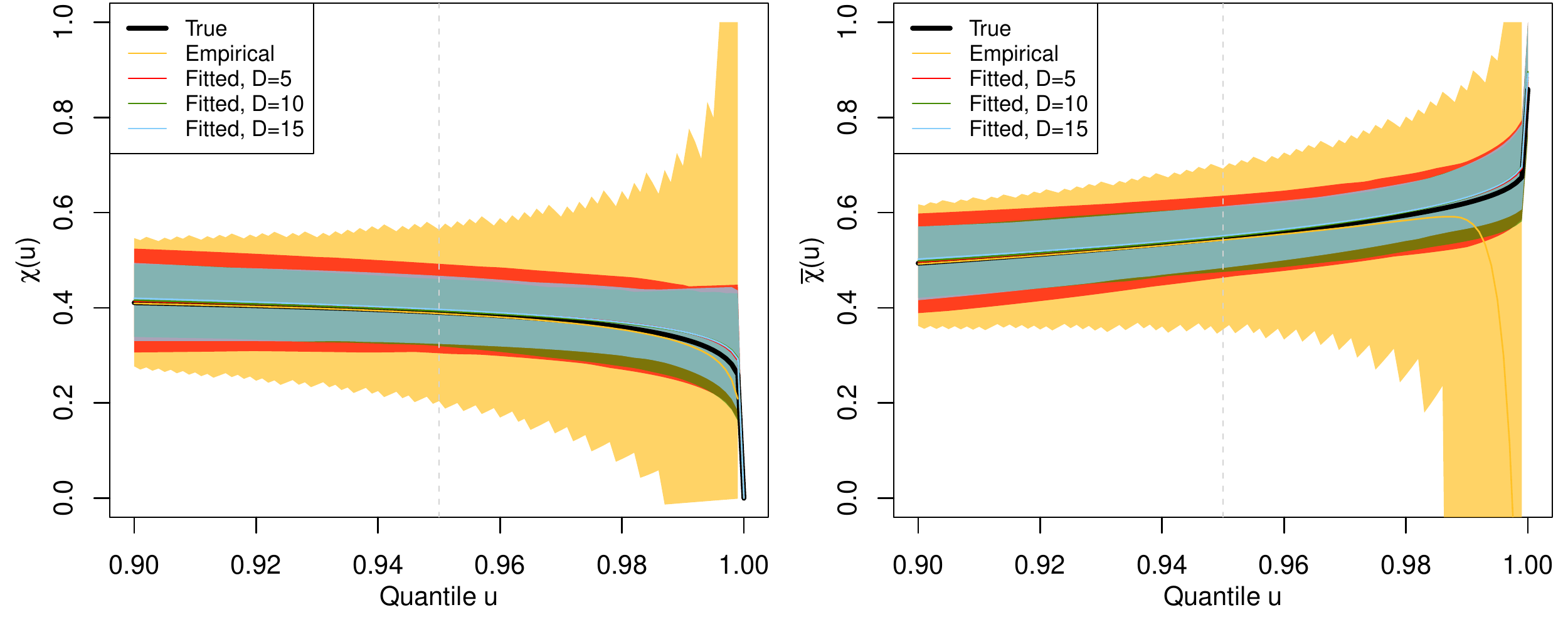}
\caption{Estimated coefficients $\chi(u)$ (left) and $\bar\chi(u)$ (right), $u\in[0.9,1]$, for model \eqref{RandomScale} using \eqref{ModelAD_AI} with  $\beta=\gamma=1$ (asymptotic independent case with $\chi=0$) and correlation function $\rho(\boldsymbol{s}_1,\boldsymbol{s}_2)=\exp\{-\|\boldsymbol{s}_2-\boldsymbol{s}_1\|/\lambda\}$ with $\lambda=1$ for two points at distance $\|\boldsymbol{s}_2-\boldsymbol{s}_1\|=0.5$. Estimation is either nonparametric (yellow) or parametric based on $D=5$ (red), $10$ (green) or $15$ (blue) uniform locations in $[0,1]^2$. The number of replicates is $n=1000$. Solid lines show means of $500$ simulations, while shaded areas are $95\%$ overall confidence envelopes. True curves are in black, and the threshold $v=0.95$ used in \eqref{eq:lik} is the vertical dashed line.}\label{fig:results2}
\end{figure}


Preliminaries about copulas are recalled in \S\ref{subsec:copulas}. The definition and properties of Gaussian scale mixtures are detailed in \S\ref{sec:GaussianScaleMix}. In \S\ref{sec:AsymptAIAD}, we unify and extend theoretical results on joint tail decay rates and detail the conditions leading to asymptotic dependence or asymptotic independence of Gaussian scale mixtures. Parametric modeling is discussed in \S\ref{InterpAIandAD}, and we propose several new models able to transition from one asymptotic regime to the other. Full likelihood inference is discussed in \S\ref{subsec:Lik}, followed by a simulation study in \S\ref{sec:SimulStudy}. We illustrate our modeling approach in \S\ref{sec:Application}  by analyzing wind speed extremes in the Pacific Northwest, US. Section~\ref{sec:Discussion} concludes with some discussion. Proofs are deferred to the Appendix.


\subsection{Copula models}\label{subsec:copulas}

By Sklar's Theorem \citep{Sklar:1959}, any continuous joint distribution $G(\boldsymbol{x})$, $\boldsymbol{x}=(x_1,\ldots,x_D)^T\in{\rm Supp}(G)\subset\Real^D$, with univariate margins $G_1,\ldots,G_D$ may be uniquely represented as
\begin{equation}\label{eq:CopulaDistribution}
G(\boldsymbol{x})=C\{G_1(x_1),\ldots,G_D(x_D)\}\qquad\Longleftrightarrow\qquad C(\boldsymbol{u})=G\{G_1^{-1}(u_1),\ldots,G_D^{-1}(u_D)\},
\end{equation}
where $C(\boldsymbol{u})$, $\boldsymbol{u}=(u_1,\ldots,u_D)^T\in[0,1]^D$, is the copula, also called the dependence function. Alternatively, a copula may be defined as a joint distribution with uniform margins on $[0,1]$. The interesting feature of the representation \eqref{eq:CopulaDistribution} is that it enables separate treatment of marginal distributions and dependence structure. Many copula families may be constructed using \eqref{eq:CopulaDistribution}: the Gaussian copula is obtained by taking $G(\cdot)=\Phi_D(\cdot;\boldsymbol{\Sigma})$, the multivariate standard Gaussian distribution with correlation matrix $\boldsymbol{\Sigma}$; the Student-$t$ copula is obtained by taking $G(\cdot)=T_D(\cdot;\boldsymbol{\Sigma},{\rm Df})$, the multivariate Student-$t$ distribution with correlation matrix $\boldsymbol{\Sigma}$ and ${\rm Df}>0$ degrees of freedom. 
In \S\ref{sec:GaussianScaleMix} we define Gaussian scale mixture models that are elliptic extensions of Gaussian processes from which flexible copula families can be derived.

\section{Gaussian scale mixture processes}\label{sec:GaussianScaleMix}

\subsection{Definition}
To create flexible spatial models, we define a Gaussian scale mixture process (i.e., a Gaussian process with random variance) as follows:
\begin{eqnarray}\label{RandomScale}
X(\boldsymbol{s})=RW(\boldsymbol{s}), \quad \boldsymbol{s}\in\calS\subset\Real^d,
\end{eqnarray}
where $W(\boldsymbol{s})$ is a standard Gaussian process with correlation function $\rho(\boldsymbol{s}_1,\boldsymbol{s}_2)$, and $R\sim F(r)$ is a positive random variable, independent of $W(\boldsymbol{s})$, with distribution $F(r)$ and density $f(r)$ if the latter exists, which we will assume in the remainder of the paper if not stated otherwise.  Conditional on $R$, the random process $X(\boldsymbol{s})$ is Gaussian with zero mean and variance $R^2$. Gaussian processes arise as a special case when $R= r_0$ almost surely for some $r_0>0$. In this paper, we use the copula associated to \eqref{RandomScale} through \eqref{eq:CopulaDistribution} as a model for extremal dependence. 

\subsection{Finite dimensional distributions}

When the process \eqref{RandomScale} is observed at $D$ spatial locations $\boldsymbol{s}_1,\ldots,\boldsymbol{s}_D\in\calS$, we write  $X_j=X(\boldsymbol{s}_j)$ and $W_j=W(\boldsymbol{s}_j)$, $j=1,\ldots,D$, yielding the random vectors $\boldsymbol{X}=(X_1,\ldots,X_D)^T$ and $\boldsymbol{W}=(W_1,\ldots,W_D)^T$. From \eqref{RandomScale}, one has the representation
\begin{eqnarray}\label{RandomScale2}
\boldsymbol{X}=R\boldsymbol{W}, \qquad R\sim F(r) \indep\boldsymbol{W}\sim\calN_D(0,\boldsymbol{\Sigma}),
\end{eqnarray}
where $\boldsymbol{\Sigma}$ is a correlation matrix determined by the spatial configuration of sites. By conditioning on $R$, we deduce that the distribution $G$ and density $g$ of $\boldsymbol{X}$ are
\begin{equation}
G(\boldsymbol{x})=\int_0^\infty \Phi_D(\boldsymbol{x}/r;\boldsymbol{\Sigma}) f(r){\rm d}r,\qquad g(\boldsymbol{x})=\int_0^\infty \phi_D(\boldsymbol{x}/r;\boldsymbol{\Sigma}) r^{-D} f(r){\rm d}r,\qquad \boldsymbol{x}\in\Real^D,\label{CDF_PDF}
\end{equation}
where $\Phi_D(\cdot;\boldsymbol{\Sigma})$ and $\phi_D(\cdot;\boldsymbol{\Sigma})$, respectively, denote the $D$-variate Gaussian distribution and density with zero mean and covariance matrix $\boldsymbol{\Sigma}$. Some non-trivial  choices of the mixing density $f(r)$ lead to a closed-form expression of the density $g(\boldsymbol{x})$, including the Student-$t$, Laplace and slash models \citep{Kotz.al.2004}.  In general, the unidimensional integrals in \eqref{CDF_PDF} can be accurately approximated using numerical integration. 
Marginal distributions $G_k$ and their corresponding densities $g_k$, $k=1,\ldots,D$, are
\begin{equation}\label{CDF_PDF_margins}
G_k(x_k)=\int_0^\infty \Phi(x_k/r) f(r){\rm d}r,\qquad g_k(x_k)=\int_0^\infty \phi(x_k/r) r^{-1} f(r){\rm d}r,\qquad x_k\in\Real,
\end{equation}
where $\Phi(\cdot)=\Phi_1(\cdot;1)$ and $\phi(\cdot)=\phi_1(\cdot;1)$ denote the univariate standard Gaussian distribution and density, respectively. The censored likelihood defined in \S\ref{subsec:Lik} requires the partial derivatives of the copula distribution and hence partial derivatives of the distribution $G$ in \eqref{CDF_PDF}. Let $I\subset\{1,\ldots,D\}$ be a set of indices of cardinality $|I|$ corresponding to components exceeding a high threshold in \S\ref{subsec:Lik}, with its complement in $\{1,\ldots,D\}$ denoted by $I^c$, and let $\boldsymbol{x}=(x_1,\ldots,x_D)^T\in\Real^D$. For any sets of indices $A,B\subset\{1,\ldots,D\}$, let $\boldsymbol{x}_A$ denote the subvector of $\boldsymbol{x}$ obtained by retaining the components indexed by $A$, let $\boldsymbol{\Sigma}_{A;B}$ denote the matrix $\boldsymbol{\Sigma}$ restricted to the rows in $A$ and the columns in $B$, and let $\boldsymbol{\Sigma}_{A\mid B}=\boldsymbol{\Sigma}_{A;A}-\boldsymbol{\Sigma}_{A;B}\boldsymbol{\Sigma}_{B;B}^{-1}\boldsymbol{\Sigma}_{B;A}$. Differentiation of $G$ in \eqref{CDF_PDF} yields
\begin{eqnarray}
G_I(\boldsymbol{x})&:=&{\partial^{|I|} \over\partial \boldsymbol{x}_I}G(\boldsymbol{x})=\int_0^\infty {\partial^{|I|} \over\partial \boldsymbol{x}_I}\Phi_D(\boldsymbol{x}/r;\boldsymbol{\Sigma}) f(r){\rm d}r\nonumber\\
&=&\int_0^\infty \Phi_{|I^c|}\left\{(\boldsymbol{x}_{I^c}-\boldsymbol{\Sigma}_{I^c;I}\boldsymbol{\Sigma}_{I;I}^{-1}\boldsymbol{x}_{I})/r;\boldsymbol{\Sigma}_{I^c\mid I}\right\} \phi_{|I|}(\boldsymbol{x}_{I}/r;\boldsymbol{\Sigma}_{I;I}) r^{-|I|} f(r){\rm d}r,\qquad\label{PartialDeriv}
\end{eqnarray}
which involves the conditional Gaussian distribution of $\boldsymbol{X}_{I^c}$ given $\boldsymbol{X}_I=\boldsymbol{x}_I/r$. The computation of $G$ in \eqref{CDF_PDF} and $G_I$ in \eqref{PartialDeriv} relies on the Gaussian distribution function in dimension $D$ and $|I^c|$ respectively, which can be estimated without bias \citep{Genz.Bretz:2009}.

\subsection{Interpretation as elliptic processes}

Gaussian scale mixtures \eqref{RandomScale2} are elliptically contoured distributions \citep{ca1981}, which may be written in pseudo-polar representation as 
\begin{equation}\label{GaussianScaleElliptic}
\boldsymbol{X} \Deq R^\star\boldsymbol{\Sigma}^{1/2}\boldsymbol{U},
\end{equation}
where $\boldsymbol{\Sigma}=\boldsymbol{\Sigma}^{1/2}\boldsymbol{\Sigma}^{T/2}$ is a covariance matrix and $R^\star\geq 0$ is a positive random variable, called radius, that is independent of a random vector $\boldsymbol{U}=(U_1,\ldots,U_D)^T$ uniformly distributed on the Euclidean unit sphere in $\Real^D$ (i.e., $\|\boldsymbol{U}\|_2=1$). We assume that $\boldsymbol{\Sigma}$ is invertible if not stated otherwise. Elliptical distributions can be viewed as a  random scaling of a uniform random vector residing on the unit sphere defined with respect to the Mahalanobis norm  $\|\boldsymbol{x}\|_{\boldsymbol{\Sigma}}=\sqrt{\boldsymbol{x}^T\boldsymbol{\Sigma}^{-1}\boldsymbol{x}}$. Using elliptic theory \citep{ca1981}, one can equivalently rewrite the multivariate density in \eqref{CDF_PDF} as
\begin{equation*}
g(\boldsymbol{x})=|\boldsymbol{\Sigma}|^{-1/2}h_D\left(\|\boldsymbol{x}\|_{\boldsymbol{\Sigma}}^2\right),\qquad \boldsymbol{x}\in\Real^D,
\end{equation*}
for some function $h_D:[0,\infty)\rightarrow[0,\infty)$. The $D$-variate Gaussian distribution is characterized by $h_D(t)=(2\pi)^{-D/2}\exp(-t/2)$, $t\geq 0$. A simple change of variables gives the density $f_{R^\star}$ of the radial component $R^\star$ in \eqref{GaussianScaleElliptic} as $f_{R^\star}(r) = A_D r^{D-1}h_D(r^2)$, $r>0$, where  $A_D$ denotes the surface area of the unit ball in $\Real^D$ (i.e., $A_1=1$ and $A_D=2\pi^{D/2}\{\Gamma(D/2)\}^{-1}$, $D>1$, with the gamma function $\Gamma(\cdot)$). Because the Gaussian vector $\boldsymbol{W}$ in \eqref{RandomScale2} is itself elliptic, the radial variable of $\boldsymbol{X}=R\boldsymbol{W}$ is $R^\star=RR_{\boldsymbol{W}}$ with $R_{\boldsymbol{W}}$ distributed according to the chi-distribution with $D$ degrees of freedom and has density 
$f_{R_{\boldsymbol{W}}}(r) = 2^{1-D/2} \{\Gamma(D/2)\}^{-1}r^{D-1}\exp(-r^2/2)$,  $r>0$.
One reason that the class of Gaussian scale mixtures provides a practically relevant family of models is that under mild restrictions, it coincides with the large family of stochastic processes possessing elliptic finite-dimensional distributions \citep{Huang.Cambanis.1979}.

\subsection{Conditional distributions and simulation algorithm}\label{sec:condsimalgo}
Conditional simulation is crucial for spatial prediction and estimation of complex spatial functionals. We demonstrate how this can be efficiently performed for spatial or multivariate models of the form \eqref{RandomScale} or \eqref{RandomScale2}, respectively. For $\boldsymbol{X}=(\boldsymbol{X}_1^T,\boldsymbol{X}_2^T)^T=R(\boldsymbol{W}_1^T,\boldsymbol{W}_2^T)^T$ a $D$-dimensional Gaussian scale mixture partitioned into subvectors $\boldsymbol{X}_1$ and $\boldsymbol{X}_2$ of dimensions $D_1\in \{1,\ldots,D-1\}$ and $D_2=D-D_1$, respectively, we derive the conditional distributions of $\boldsymbol{X}_2$ given $\boldsymbol{X}_1$  and of the latent variable $R$ given $\boldsymbol{X}_1$. We then use these results to propose a conditional simulation algorithm. We let $\boldsymbol{\Sigma}_{i;j}$, $i,j\in\{1,2\}$, denote the corresponding blocks of the covariance matrix $\boldsymbol{\Sigma}$ of $\boldsymbol{W}=(\boldsymbol{W}_1^T,\boldsymbol{W}_2^T)^T$ and $\boldsymbol{\Sigma}_{i\mid j}=\boldsymbol{\Sigma}_{i;i}-\boldsymbol{\Sigma}_{i;j}\boldsymbol{\Sigma}_{j;j}^{-1}\boldsymbol{\Sigma}_{j;i}$.

\begin{theorem}[Conditional distributions]
\label{theor:cond}
The conditional distribution of $\boldsymbol{X}_2$ given $\boldsymbol{X}_1= \boldsymbol{x}_1$ is elliptic with density
\begin{equation}\label{condsim1}
f_{ \boldsymbol{X}_2\mid  \boldsymbol{X}_1= \boldsymbol{x}_1}(\boldsymbol{x}_2) = c_0^{-1}|\boldsymbol{\Sigma}_{2\mid 1}|^{-1/2}h_D\left\{ (\boldsymbol{x}_2 - \boldsymbol{\mu}_{2\mid 1})^T \boldsymbol{\Sigma}_{2\mid 1}^{-1}  (\boldsymbol{x}_2 - \boldsymbol{\mu}_{2\mid 1}) + c_1\right\},\qquad \boldsymbol{x}_2\in\Real^{D_2},
\end{equation}
where $\boldsymbol{\mu}_{2\mid 1} = \boldsymbol{\Sigma}_{2;1}\boldsymbol{\Sigma}_{1;1}^{-1} \boldsymbol{x}_1$, $c_0 =A_{D_2}\int_0^\infty h_D(r^2+c_1)r^{D_2-1} \mathrm{d}r$, $c_1= \boldsymbol{x}_1^T\boldsymbol{\Sigma}_{1;1}^{-1} \boldsymbol{x}_1$ and
\begin{equation*}
h_D(t) =A_D^{-1}t^{(1-D)/2}f_{RR_{\boldsymbol{W}}}(\sqrt{r}), \qquad f_{RR_{\boldsymbol{W}}}(r) =\int_0^\infty  F(s^{-1})\left\{ f_{R_{\boldsymbol{W}}}(rs) +f'_{R_{\boldsymbol{W}}}(rs)rs\right\}\mathrm{d}s.
\end{equation*}
It has pseudo-polar representation $\boldsymbol{\mu}_{2\mid 1}+R_{2\mid 1}\boldsymbol{\Sigma}_{2\mid 1}^{1/2} \boldsymbol{U}$ with $\boldsymbol{\Sigma}_{2\mid 1}=\boldsymbol{\Sigma}_{2\mid 1}^{1/2}\boldsymbol{\Sigma}_{2\mid 1}^{T/2}$, $\|\boldsymbol{U}\|_2=1$  and radius $R_{2\mid 1}$ whose density is $f_{R_{2\mid 1}}(r) = A_{D_2}\,c_0^{-1} r^{D_2-1} h_D(r^2+c_1)$, $r > 0$.
The conditional density of $R$ given $\boldsymbol{X}_1=\boldsymbol{x}_1\in\Real^{D_1}$, with $1\leq D_1\leq D$, is
\begin{equation}\label{condsim2}
f_{R\mid  \boldsymbol{X}_1= \boldsymbol{x}_1}(r) = r^{-D_1}f(r)\phi_{D_1}(\boldsymbol{x}_1/r; \boldsymbol{\Sigma}_{1;1})/g(\boldsymbol{x}_1),\qquad r\geq0.
\end{equation}
\end{theorem}

Simulation of $\boldsymbol{X}_2$ conditional on $\boldsymbol{X}_1= \boldsymbol{x}_1$ can be done either by directly calculating and simulating the elements of the pseudo-polar representation $\boldsymbol{\mu}_{2\mid 1}+R_{2\mid 1}\boldsymbol{\Sigma}_{2\mid 1}^{1/2} \boldsymbol{U}$ or by exploiting the latent Gaussian structure in a two-step procedure: to simulate $R\boldsymbol{W}_2$ conditional on $R\boldsymbol{W}_1=\boldsymbol{x}_1$, we first generate a realization $\tilde{r}$ of the conditional scale variable $\tilde{R}$ according to its density $f_{R\mid \boldsymbol{X}_1= \boldsymbol{x}_1}$ in \eqref{condsim2}, and then we sample a realization $\tilde{\boldsymbol{w}}_2$ of  $\boldsymbol{W}_2$ conditional on $\tilde{R}=\tilde{r}$ and $\boldsymbol{X}_1=\boldsymbol{x}_1$, i.e., we sample $\tilde{\boldsymbol{w}}_2$  according to the conditional Gaussian distribution $\boldsymbol{W}_2 \mid  \boldsymbol{W}_1 =  \boldsymbol{x}_1/\tilde{r}$ with mean $\boldsymbol{\mu}_{2\mid 1}/\tilde{r}$ and covariance matrix $\boldsymbol{\Sigma}_{2\mid 1}$. Then,  $\tilde{r}\tilde{\boldsymbol{w}}_2$ is a realization of the conditional vector $\boldsymbol{X}_2$ given $\boldsymbol{X}_1=\boldsymbol{x}_1$.

\begin{figure}[t!]
\centering
\includegraphics[width=\linewidth]{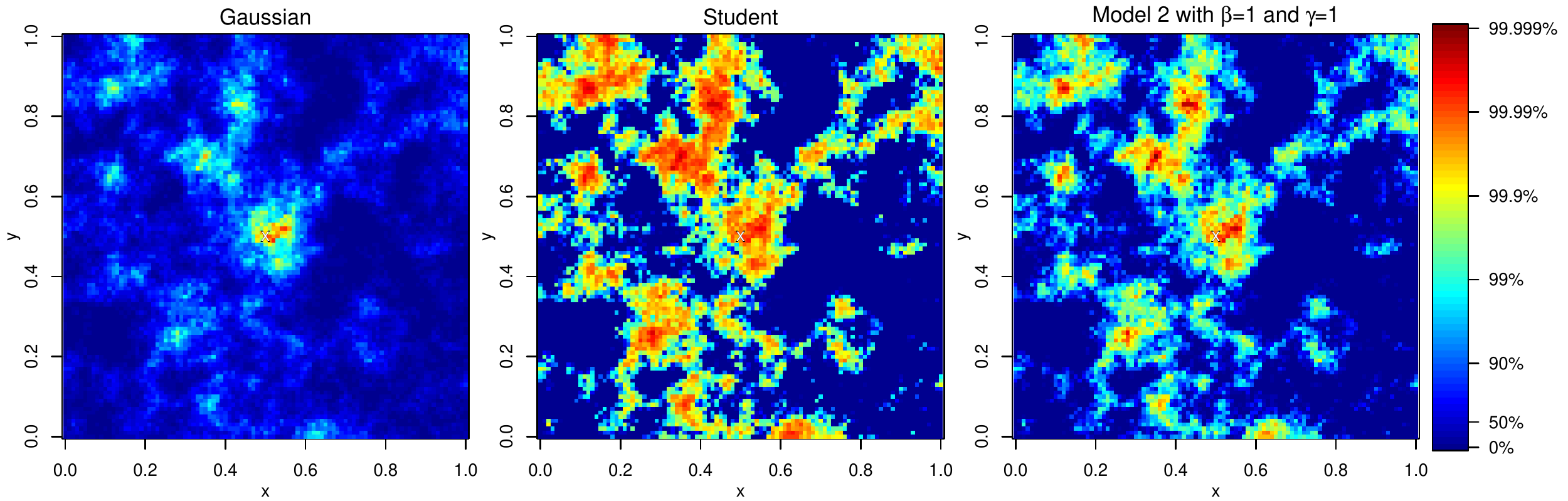}
\caption{Conditional simulations in $[0,1]^2$ from an (asymptotically independent) Gaussian copula (left), an (asymptotically dependent) Student-$t$ copula with $3$ degrees of freedom (middle) and one of our new proposed Gaussian scale mixtures defined in \eqref{ModelAD_AI} with $\beta=\gamma=1$ leading to asymptotic independence (right), displayed on exponential margins. The underlying Gaussian field has exponential correlation function $\rho(\boldsymbol{s}_1,\boldsymbol{s}_2)=\exp(-\|\boldsymbol{s}_1-\boldsymbol{s}_2\|/0.1)$. Simulations are done conditionally on the central grid cell at $(0.5,0.5)$ being equal to the $99.99\%$-quantile. The conditional simulation of the scale variable $\tilde{R}$ is based on a Metropolis-Hastings algorithm with multiplicative random walk, and the three conditional simulations of the Gaussian component $\tilde{\boldsymbol{W}}$ use the same random seed.}
\label{fig:condsims}
\end{figure}

To illustrate our conditional simulation algorithm based on the latent Gaussian structure, Figure~\ref{fig:condsims} displays realizations of three Gaussian scale mixture models on $[0,1]^2$, each of them conditioned on being equal to the $99.99\%$-quantile at the point $(0.5,0.5)$. By using the same random seed for the conditional Gaussian sample, the impact of choosing different random scale distributions emerges clearly. Large contiguous areas with high values arise for the asymptotically dependent Student-$t$ model, while the Gaussian one generates spatially isolated and highly localized peaks with highly variable peak size.  The third simulation exemplifies our new proposed model \eqref{ModelAD_AI} providing a continuous transition from Gaussian to asymptotic dependence, here with parameters leading to asymptotic independence; more details are provided in \S\ref{InterpAIandAD}. For a moderate number of conditioning values (here $D_1=1$ and $D_2=10200$), simulation times are largely dominated by the conditional Gaussian simulation, while the computational effort for simulating the conditional random scale variable is negligible.
 
\section{Tail behavior of Gaussian scale mixture models}\label{sec:Asympt}
\subsection{Asymptotic independence and dependence}\label{sec:AsymptAIAD}
We now characterize the bivariate joint tail decay of Gaussian scale mixtures with two commonly used coefficients $\chi$ and $\bar{\chi}$ \citep{Coles.etal:1999} defined as $\chi:=\lim_{u\to 1} \chi(u)$ and $\bar\chi:=\lim_{u\to 1} \bar\chi(u)$, where 
\begin{equation}\label{eq:chi_chibar}
\chi(u)=2- {\log C(u,u)\over \log(u)},\qquad \bar\chi(u)={2\log (1-u)\over \log \bar C(u,u)}-1,
\end{equation}
and $C(u_1,u_2)$ is the bivariate copula stemming from the pair of variables $(X_1,X_2)^T$ through~\eqref{eq:CopulaDistribution} with survival copula $\bar C(u_1,u_2)=1-u_1-u_2+C(u_1,u_2)$. When variables $X_1$ and $X_2$ are asymptotically dependent, then $\chi>0$ and $\bar{\chi}=1$. By contrast, asymptotic independence of $X_1$ and $X_2$ gives $\bar{\chi}\in [-1,1]$ and $\chi=0$. In practice, the pair of coefficients $\{\chi(u),\bar{\chi}(u)\}$ evaluated for increasingly large thresholds $u\in[0,1]$ may help assess the asymptotic dependence class \citep{Coles.etal:1999}; see  Figure~\ref{fig:results2} for illustration.

The following new results on joint tail decay rates rely on specific tail dependence characterizations of elliptic distributions \citep[see the seminal monograph of][]{Berman:1992}. To understand the asymptotic dependence structure entailed by general Gaussian scale mixture models \eqref{RandomScale2}, we fix $D=2$ and study the asymptotic properties of $\boldsymbol{X}=(X_1,X_2)^T=R(W_1,W_2)^T$, where the Gaussian vector $\boldsymbol{W}=(W_1,W_2)^T$ has correlation $\rho\in (-1,1)$, excluding the degenerate cases $\rho\in\{-1,1\}$.  We first provide general results covering asymptotic dependence and asymptotic independence and then discuss useful parametric examples of distributions for $R$ in \ref{InterpAIandAD}.  Interestingly, asymptotic independence is obtained for the wide class of Weibull-like tail decay in $R$ (see Theorem~\ref{theor:tailindep}), whereas asymptotic dependence occurs when $R$ is regularly varying at infinity, i.e., when $R$ has Pareto-like tail behavior (see Theorem~\ref{theor:taildep}). 

\begin{theorem}[Asymptotic independence for Gaussian scale mixtures] \label{theor:tailindep}\ 
Suppose that 
\begin{equation}\label{eq:AIV}
\pr(R\geq r) = 1-F(r) \sim \alpha r^{\gamma} \exp(-\delta r^\beta), \qquad r\to\infty,
\end{equation}
for some constants $\alpha>0$, $\beta>0$, $\gamma\in \Real$ and   $\delta>0$. Then $\chi = 0$ and 
\begin{equation*}
\bar{\chi}= 2 \left\{(1+\rho)/ 2\right\}^{\beta/(\beta+2)} - 1.
\end{equation*}
The joint tail can be written as\begin{equation}\label{eq:AIbv}
\bar C\{1-1/x,1-1/x\} = \calL(x)x^{-1/\eta}, \qquad  x\rightarrow \infty,
\end{equation}
where $\eta=(1+\bar{\chi})/2$ is the coefficient of tail dependence \citep{Ledford.Tawn:1996}, $\calL(x)\sim K \log(x)^{(1-1/\eta){2\gamma+\beta\over2\beta}+1/(2\eta)-1}$ is a slowly varying function as $x\rightarrow\infty$ and $K$ is a positive constant depending on  $\alpha$, $\beta$, $\gamma$ and $\delta$; see the proof in the appendix.
\end{theorem}

\begin{theorem}[Asymptotic dependence for Gaussian scale mixtures] \label{theor:taildep}\ 
Suppose that $R$ is regularly varying at infinity, that is, 
\begin{equation}\label{eq:ADV}
{\pr(R\geq tr)\over \pr(R \geq t)}  = {1-F(tr)\over 1-F(t)} =  r^{-\gamma}, \quad  r > 0, \quad t\rightarrow \infty, 
\end{equation}
for some $\gamma>0$. 
Then $\bar{\chi} = 1$ and 
\begin{equation}\label{eq:ADchi}
\chi =  2\left[1-T\left\{(1+\gamma)^{1/2}(1-\rho)(1-\rho^2)^{-1/2};\gamma+1\right\}\right],
\end{equation}
where $T(\cdot;{\rm Df})=T_1(\cdot;1,{\rm Df})$ is the univariate Student-$t$ distribution with ${\rm Df}>0$ degrees of freedom. 
The joint tail can be written as
\begin{equation}\label{eq:ADbv}
\bar C(1-1/x,1-1/x) \sim \chi\times \pr\left\{G_1(X_1)>1-1/x\right\} \sim \chi/x, \quad x\rightarrow \infty.
\end{equation}
\end{theorem}
We use the terms Weibull-type and Pareto-type distributions for variables $R$ with tail representation \eqref{eq:AIV} and \eqref{eq:ADV}, respectively. 
The case where $R$ is deterministic or upper-bounded almost surely can be interpreted as a limit of \eqref{eq:AIV} as $\beta\rightarrow\infty$ and in this case $\bar{\chi}=\rho$.  More general results on off-diagonal decay rates are in the Supplementary Material.



To better understand the extremal dependence structure of Gaussian scale mixture processes, we shortly recall related max-stable limits; technical details are given in the Supplementary Material.  For a regularly varying scale distribution $F$ as defined in \eqref{eq:ADV} we get  extremal-$t$ limit processes \citep{op2013}. By contrast, if $F$ has a Weibull-type tail as defined by \eqref{eq:AIV} and provided that the Gaussian correlation is not perfect between distinct sites, the asymptotic independent structure yields a ``white noise'' max-stable limit. In this case, more insight can be obtained by considering triangular arrays of Gaussian scale mixtures with correlation increasing to one  at a certain speed \citep{Hashorva.2013}, similar to standard results for Gaussian triangular arrays \citep{Husler.Reiss:1989}, resulting in Brown--Resnick limit processes \citep{Kabluchko.al.2009}.  Brown--Resnick processes further arise as certain limits of extremal-$t$ processes, in such a way that appropriately defined triangular arrays of Gaussian scale mixtures with regularly varying scale variable converge to Brown--Resnick limits when the regular variation index $\gamma$ in \eqref{eq:ADV} tends to $\infty$.

\subsection{Models bridging asymptotic dependence and independence}\label{InterpAIandAD}
When modeling dependence in spatial extremes, fixing the type of asymptotic behavior has important consequences on the estimation of return levels for spatial functionals. 
This choice may be guided by a preliminary assessment of asymptotic dependence using the coefficients $\chi(u)$ and $\bar\chi(u)$ defined in \eqref{eq:chi_chibar}, although nonparametric estimates may be highly variable  especially with small sample sizes as uncertainties become increasingly large as $u\to1$; recall Figure~\ref{fig:results2}. Thus, preliminary selection of the asymptotic dependence regime may be awkward, and it is highly valuable to fit flexible models encompassing both  regimes in order to borrow strength across all stations for efficient estimation of tail properties from the data. As Gaussian processes have been widely applied and advocated in spatial statistics, it is useful to have the Gaussian dependence structure as a special case. Gaussian models may provide an appropriate description of the dependence in some applications, and it is always instructive to assess how far the fitted dependence model is from Gaussianity. 
Based on these arguments, a useful distribution for $R$ has to be chosen when fitting a Gaussian mixture model to observed extremes. We first describe some known parametric families that may be used to capture either asymptotic independence or asymptotic dependence, then we propose novel parsimonious models that encompass the two dependence classes.

In the asymptotic independence case, a flexible distribution satisfying Weibull-type tail behavior as in \eqref{eq:AIV} is the generalized gamma distribution; it covers the full range of Weibull coefficients $\beta>0$ in \eqref{eq:AIV}  and encompasses several well-known simpler parametric families, including the gamma distribution, the Weibull, exponential, Nakagami, Rayleigh, chi, chi-squared, half-normal and log-normal ones; see the Supplementary Material for details. However, its three parameters may lead to issues of parameter parsimony and identifiability in practice. A Rayleigh-distributed random scale $R$ yields Laplace random fields $X(\boldsymbol{s})$ with explicit joint densities $g(\boldsymbol{x})$ in \eqref{CDF_PDF} and a certain type of univariate generalized Pareto tails \citep{op2015}. Asymptotic dependence models are obtained from a regularly varying random scale $R$; recall Theorem~\ref{theor:taildep}. 
One possibility is to choose the distribution of $R^2$ as the inverse-gamma distribution with scale parameter $2$ and shape parameter ${\rm Df}/2$, yielding Student-$t$ random fields with ${\rm Df}>0$ degrees of freedom \citep{ro2006,ma2005} and converging to the Gaussian limit as ${\rm Df}\to\infty$. 
When $R$ is Pareto distributed, the Gaussian scale mixture has multivariate slash densities available in closed form, albeit in terms of the incomplete gamma function \citep{wa2006}. 

We now explore models able to bridge the two asymptotic regimes while keeping flexibility in both of the respective submodels. To assess the bivariate tail flexibility of Gaussian scale mixtures, we consider the range of possible $\chi$ and $\overline{\chi}$ coefficients that may be generated for a fixed correlation coefficient value $-1<\rho<1$ in the underlying bivariate Gaussian variable. Fixing $\rho$ in the Gaussian for comparing tail dependence is appropriate since it means fixing  $\rho$ in all of its scale mixtures (provided second moments exist) and thus fixing Kendall's $\tau$. Among the aforementioned models, only the (asymptotically dependent) Student-$t$ one covers both regimes by considering its extension to the (asympotically independent) Gaussian model on the boundary of its parameter space when $\mathrm{Df}=\infty$. Theorem \ref{theor:taildep} shows that any value of $\chi\in[0,1]$ may be obtained by varying the degrees of freedom ${\rm Df}$, with $\chi\downarrow 0$ as $\mathrm{Df}\rightarrow \infty$, which implies that the Student-$t$ model is fairly flexible in the asymptotic dependence case. However, thanks to Theorem~\ref{theor:tailindep}, only two values of $\overline{\chi}$ can be obtained, namely $\overline{\chi}=\rho$ when $\mathrm{Df}=\infty$ or $\overline{\chi}=1$ when $\mathrm{Df}<\infty$. The jump in the $\overline{\chi}$ value when moving from asymptotic dependence to asymptotic independence implies a lack of flexibility to capture joint tail decay rates in the asymptotic independence regime, owing to the relatively rigid Gaussian tail structure. We now propose three new random scale models, which provide a smoother transition between asymptotic dependence classes and/or provide more flexibility to detect and model the asymptotic dependence regime in data.

For Model~1, we transpose  the arguments of \citet{Wadsworth.al.2017} concerning the choice of the random scaling in pseudo-polar representations to our framework, and we model the random scale through  a generalized Pareto distribution, 
	\begin{equation}\label{ModelAD_AI3}
	F(r)= \begin{cases}
	1-(1+ \xi r)_+^{-1/\xi},&\xi\neq 0, \\
	1-\exp(-r),&\xi=0,
	\end{cases},\qquad 0 \leq r < r^\star, 
	\end{equation}
	with upper endpoint $r^\star=\infty$ when $\xi\geq 0$ and $r^\star=-1/\xi$ otherwise, and where $(x)_+=\max(x,0)$.
   Model~1 is Pareto-like with tail parameter $\gamma:=1/\xi$ in \eqref{eq:ADV} when $\xi>0$, yielding asymptotic dependence thanks to Theorem~\ref{theor:taildep} with $\overline{\chi}=1$ and $\chi\downarrow 0$ as $\xi\downarrow 0$.  It is exponential (i.e., Weibull-like with tail parameter $\beta=1$ in \eqref{eq:AIV}) when $\xi=0$ or $\xi\to0$, yielding asymptotic independence with $\overline{\chi}=2\{(1+\rho)/2\}^{1/3}-1$ thanks to Theorem~\ref{theor:tailindep}. Finally, it is upper-bounded when $\xi<0$, which gives  asymptotic independence with $\overline{\chi}=\rho$. Model~1 is a rather flexible model for asymptotic dependence similar to the Student-$t$ one, and the transition to asymptotic independence takes place in the interior of its parameter space, but it still lacks flexibility for asymptotic independence as only two specific values of $\bar{\chi}<1$ can arise for fixed $\rho$. Therefore, both the Student-$t$ model and Model~1 should be used when there is a dominating yet still uncertain  suspicion that the data are asymptotically dependent.

Since we seek to obtain more flexiblity and a smoother transition to asymptotic dependence within asymptotic independence, we propose a new Model~2 that can generate any value of $\bar\chi\in[\rho,1]$, for fixed Gaussian correlation $\rho<1$. This novel  two-parameter Weibull-type distribution with support $[1,\infty)$ contains the Dirac mass at $1$ as limiting case, yielding asymptotically independent standard Gaussian processes, and the Pareto distribution as boundary case, yielding asymptotic dependence; recall Theorem~\ref{theor:taildep}. Its distribution with parameters $\beta\geq0$ and $\gamma>0$ is defined as 
\begin{equation}\label{ModelAD_AI}
F(r) = \begin{cases}
1 - \exp\left\{-\gamma  (r^{\beta}-1)/\beta\right\}, &\beta>0,\\
1 - r^{-\gamma}, &\beta=0,
\end{cases}\quad r \geq 1.
\end{equation}
The distribution \eqref{ModelAD_AI} forms a continuous parametric family with respect to $\beta$ as the term $(r^{\beta}  -1)/\beta$ converges to $\log r$ as $\beta\downarrow0$. 
The type of asymptotic dependence is determined by the value of $\beta$. When $\beta>0$, \eqref{ModelAD_AI} coincides with the tail representation \eqref{eq:AIV} with the same tail parameter $\beta$, yielding asymptotic independence. When $\beta=0$ or $\beta\downarrow0$, the variable $R$ is Pareto distributed with $F(r)=1-r^{-\gamma}$, $r\geq1$, yielding asymptotic dependence. 
The Dirac mass at $1$ is obtained as $\beta\rightarrow\infty$ or as $\gamma\rightarrow\infty$. The benefit of Model~2 is to provide a smooth transition from asymptotic independence to asymptotic dependence with  $\overline{\chi}\uparrow 1$ for $\beta\downarrow 0$ and $\gamma>0$ fixed; moreover, it still keeps a smooth transition from asymptotic dependence to asymptotic independence with $\overline{\chi}\downarrow 0$ as $\gamma\downarrow 0$ and $\beta=0$ fixed, leading to a Gaussian limit in analogy to the Student model. Model 2 has similarities with univariate tail models such as the extended generalized Pareto distributions proposed by \citet{Papastathopoulos.Tawn:2013} and \citet{Naveau.etal:2016}, which improve flexibility over the classical limit distribution by adding one or more extra parameters. 

A noteworthy link between Models 1 and 2 appears through the Box--Cox transform: the random scale in Model~1 arises from applying the Box--Cox transform with power $\xi$ to a standard Pareto random variable, while Model~2 is the result of applying the inverse Box--Cox transform with power $\beta$ to an exponential  variable with rate $\gamma$.


Model~2 provides high flexibility in both the asymptotic independence and asymptotic dependence cases where the latter lies on the boundary of the parameter space of the former. In practice, one may prefer to use Model~2 when there is a dominating (yet still uncertain) suspicion that the data are asymptotically independent; nonetheless, selection between separately fitted submodels through information criteria like \textsc{AIC} remains feasible. When the asymptotic dependence class is totally uncertain, it may be preferable to use a model for which the transition takes place in the interior of the parameter space, thus facilitating inference. One possibility (Model~3) is to choose a random scale variable $R$ with distribution
	\begin{equation}\label{ModelAD_AI2}
F(r)= \begin{cases}
1-(r^\star_{\beta}r)^\beta\exp[-\{(r^\star_{\beta}r)^\beta-1\}/\beta],& \beta\neq0,\\
1-r^{-1},&\beta=0,
\end{cases}\qquad r\geq  1 
\end{equation}
where $r^\star_{\beta}=\sup\{r: r^\beta\exp\{-(r^\beta-1)/\beta\}=1\}\in [1,2)$ and $\beta\in\Real$. This model is closely related to our proposal in \eqref{ModelAD_AI} and is also more parsimonious, but it may be more tricky to handle computationally as its support depends on $\beta$. From \eqref{ModelAD_AI2}, one can see that $F(r)$ is a Weibull-like distribution when $\beta>0$ with the same tail parameter $\beta$ as in \eqref{eq:AIV}, whereas $F(r)$ is a Pareto-like distribution with tail parameter $\gamma:=-\beta$ in \eqref{eq:ADV} when $\beta<0$. Furthermore, the tail dependence strength decreases as $|\beta|$ increases, irrespective of its sign. From this, identifiability issues may arise with two local likelihood maxima for $\beta>0$ and $\beta<0$, which might be bypassed by maximizing the likelihood separately with two distinct initial values.

Figure~\ref{fig:chi_chibar} shows the coefficients $\chi(u)$ and $\bar\chi(u)$ with respect to the parameters  $\beta$ and $\gamma$ of Model~2 in \eqref{ModelAD_AI}  as a function of the threshold $u\in[0.9,1]$ for various parameters and correlation $\rho$ of $(W_1,W_2)^T$. The plots illustrate the ability of Model~2 to smoothly interpolate from asymptotic independence to asymptotic dependence. To contrast the flexibility of this model \eqref{ModelAD_AI} with respect to the Gaussian copula, we also display $\chi(u)$ and $\bar\chi(u)$ for the Gaussian copula, whose correlation coefficient is chosen such that these coefficients for the two models match at the level $u=0.95$. Significant differences appear between the Gaussian copula and Model  2 when $\beta,\gamma\leq 1$, especially for $u\approx1$. When $\beta$ or $\gamma$ increases, our model approaches the Gaussian copula although the limit quantities $\bar\chi$ may still be quite different for moderate values of $\beta$ or $\gamma$. The extra tail flexibility of our model compared to the Gaussian copula is apparent when we vary the parameters $\beta$ and $\gamma$. Furthermore, the differences in joint tail decay rates with respect to the Gaussian copula are even more evident in higher dimensions.

\begin{figure}[t!]
\centering
\includegraphics[width=\linewidth]{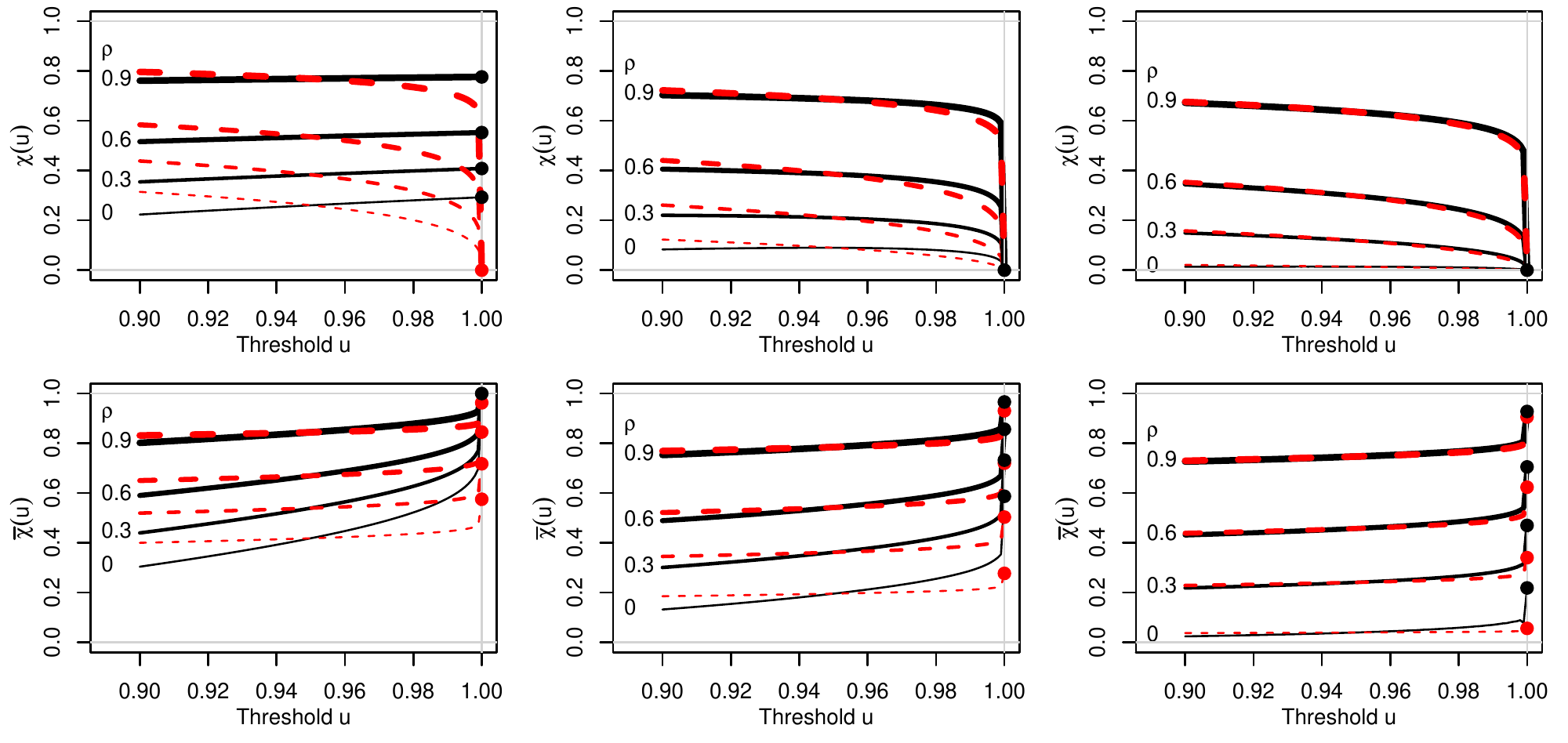}
\caption{Coefficients $\chi(u)$ (top) and $\bar\chi(u)$ (bottom), $u\in[0.9,1]$, for Model~2 defined in  \eqref{ModelAD_AI} (solid black) and for the Gaussian copula matching at $u=0.95$ (dashed red). Parameter configurations are  $\beta=0,\gamma=1$ (left; asymptotically equivalent to $\xi=1$ in Model~1 and to $\beta=0$ in Model~3), $\beta=1,\gamma=1$ (middle; asymptotically equivalent to $\xi=0$ in Model~1 and to $\beta=1$ in Model~3), $\beta=5,\gamma=1$ (right). Thin to thick curves correspond to increasing correlations $\rho=0,0.3,0.6,0.9$ for model \eqref{ModelAD_AI}. Limit quantities $\chi$ and $\bar\chi$ are represented by dots at $u=1$. The left column displays an asymptotic dependent scenario, while the middle and right columns display asymptotic independence.}\label{fig:chi_chibar}
\end{figure}
\section{Likelihood inference}\label{sec:Inference}

\subsection{Maximum likelihood approach with partial censoring}\label{subsec:Lik}


To estimate the extremal dependence structure from the observed high spatial threshold exceedances, we advocate a two-step procedure:  marginal distributions are estimated nonparametrically based on ranks and then copula parameters are estimated using a full pseudo-likelihood with partial censoring preventing estimates from being influenced by low and moderate values. Although it has never been applied to Gaussian scale mixture models, this censored approach is now quite popular in statistics of extremes \citep[see, e.g.,][]{Thibaud.etal:2013,Huser.Davison:2014a,Wadsworth.Tawn:2014}, and results from \citet{Thibaud.Opitz:2015} and \citet{Huser.etal:2016} suggest it provides a reasonable compromise between bias and variance compared to alternative approaches to fit threshold exceedances. 

Let $\boldsymbol{Y}_i=(Y_{1i},\ldots,Y_{Di})^T$, $i=1,\ldots,n$, denote $n$ independent and identically distributed observations from a process $Y(\boldsymbol{s})$ at the stations $\boldsymbol{s}_1,\ldots,\boldsymbol{s}_D\in\calS$. We assume that in the joint tail region corresponding to large values of $Y_{1i},\ldots,Y_{Di}$, the vectors $\boldsymbol{Y}_i$, $i=1,\ldots,n$, are well described by a continuous joint distribution $H$ with margins $H_1,\ldots,H_D$ and copula $C$ stemming from a Gaussian scale mixture \eqref{RandomScale}. 

Since we focus on the extremal dependence structure, we estimate marginal distributions $H_1,\ldots,H_D$ in a first step using empirical distribution functions. Defining $\hat H_k(y)={(n+1)^{-1}\sum_{i=1}^n I(Y_{ki}\leq y)}$ with the indicator function $I(\cdot)$, we transform the data to a pseudo-uniform scale as
\begin{equation}\label{transformation}
U_{ki}:=\hat H_k(Y_{ki})={{\rm rank}(Y_{ki})\over n+1}, \qquad k=1,\ldots,D,\ i=1,\ldots,n,
\end{equation}
where ${\rm rank}(Y_{ki})$ is the rank of $Y_{ki}$ among the variables $Y_{k1},\ldots,Y_{kn}$. The denominator $(n+1)$ in \eqref{transformation} ensures that transformed variables $U_{ki}$, $k=1,\ldots,D,i=1,\ldots,n$, are within $(0,1)$. Since $\hat H_k$ is a consistent estimator of $H_k$ for each component $k=1,\ldots,D$ as $n\to\infty$, the variables $U_{k1},\ldots,U_{kn}$ form an approximate uniform ${\rm Unif}(0,1)$ random sample for large $n$. The nonparametric estimator \eqref{transformation} may not be very good for the most extreme values very close to $1$, but it is very robust and its bias is of the order $\calO(n^{-1})$. We have verified through simulations that it yields reasonable results, even for moderate values of $n$; see \S\ref{sec:SimulStudy}.

In the second step, we assume that the transformed variables $U_{k1},\ldots,U_{kn}$, $k=1,\ldots,D$, are perfect random samples from the ${\rm Unif}(0,1)$ distribution, and we fit the copula \eqref{eq:CopulaDistribution} stemming from \eqref{RandomScale} (e.g., based on our model \eqref{ModelAD_AI}) by maximum likelihood, censoring values not exceeding marginal thresholds $v_1,\ldots,v_D\in(0,1)$. Let $C(\cdot;\boldsymbol{\psi})$ and $c(\cdot;\boldsymbol{\psi})$ respectively denote the chosen parametric copula family and its density, where $\boldsymbol{\psi}\in\Psi\subset\Real^p$ is the unknown vector of parameters to be estimated. The thresholds $v_k$, $k=1,\ldots,D$, are usually chosen such that $1-v_k$ is a low exceedance probability like $0.05$ or $0.01$. For each independent temporal replicate $i=1,\ldots,n$, the vector $\boldsymbol{U}_i=(U_{1i},\ldots, U_{Di})^T$ is either below, above or partially above the threshold $\boldsymbol{v}=(v_1,\ldots,v_D)^T$, leading to different censored likelihood contributions. We adopt the lowercase notation $\boldsymbol{u}_i=(u_{1i},\ldots,u_{Di})^T$, $i=1,\ldots,n$, for the realized values of $\boldsymbol{U}_i$ and write $\boldsymbol{u}_i^\star=\{\max(u_{1i},v_1),\ldots,\max(u_{Di},v_D)\}^T$. Three distinct scenarios can occur:
\begin{enumerate}
\item if all components of the vector $\boldsymbol{u}_i$ are below the threshold $\boldsymbol{v}$ (i.e., $\boldsymbol{u}_i^\star=\boldsymbol{v}$), we use the fully censored likelihood contribution
\begin{equation*}
L(\boldsymbol{u}_i;\boldsymbol{\psi}) = C(\boldsymbol{u}_i^\star;\boldsymbol{\psi}) = C(\boldsymbol{v};\boldsymbol{\psi})=G\{G_1^{-1}(v_1),\ldots,G_D^{-1}(v_D)\}.
\end{equation*}
\item if all components of the vector $\boldsymbol{u}_i$ are above the threshold $\boldsymbol{v}$ (i.e., $\boldsymbol{u}_i^\star=\boldsymbol{u}_i$),  we use the uncensored contribution 
\begin{equation*}
L(\boldsymbol{u}_i;\boldsymbol{\psi}) = c(\boldsymbol{u}_i^\star;\boldsymbol{\psi}) = c(\boldsymbol{u}_i;\boldsymbol{\psi})=\dfrac{g\{G_1^{-1}(u_{1i}),\ldots,G_D^{-1}(u_{Di})\}}{\prod_{k=1}^D g_k\{G_k^{-1}(u_{ki})\}}.
\end{equation*}
\item if some but not all components of the vector $\boldsymbol{u}_i$ are above the threshold $\boldsymbol{v}$ (indexed by the set $I_i\subset\{1,\ldots,D\}$ with complement $I_i^c\subset\{1,\ldots,D\}$), 
we use the partially censored likelihood contribution
\begin{equation*}
L(\boldsymbol{u}_i;\boldsymbol{\psi}) = C_{I_i}(\boldsymbol{u}_i^\star;\boldsymbol{\psi})= \int_{\boldsymbol{0}}^{\boldsymbol{v}_{I_i^c}} c(\boldsymbol{u}_i;\boldsymbol{\psi}) {\rm d}\boldsymbol{u}_{i_{I_i^c}}=\dfrac{G_{I_i}\{G_1^{-1}(u_{1i}^\star),\ldots,G_D^{-1}(u_{Di}^\star)\}}{\prod_{k\in I_i} g_k\{G_k^{-1}(u_{ki}^\star)\}}.
\end{equation*}
\end{enumerate}
In the above expressions, the joint distribution $G$ and density $g$, the marginal distributions $G_k$ and densities $g_k$, $k=1,\ldots,D$ and the partial derivatives $G_{I_i}$ are given by \eqref{CDF_PDF}, \eqref{CDF_PDF_margins} and \eqref{PartialDeriv} respectively. Our censored log likelihood is defined as the sum of all individual log contributions, that is,
\begin{equation}\label{eq:lik}
\ell(\boldsymbol{\psi})=\sum_{i=1}^n \log\{L(\boldsymbol{u}_i;\boldsymbol{\psi})\}.
\end{equation}
The maximum likelihood estimator (MLE) $\hat{\boldsymbol{\psi}}$, obtained by maximizing \eqref{eq:lik} over $\Psi$, is a full likelihood estimator for the censored observations $\boldsymbol{U}_i^\star=\max(\boldsymbol{U}_i,\boldsymbol{v})$ augmented by the indicator variables $I(U_{ki}\leq v_k)$, $k=1,\ldots,D$, $i=1,\ldots,n$. If the copula $C$ is well specified and if the marginal estimation performed in the first step is perfect such that the transformed observations $\boldsymbol{U}_i$ are perfectly uniform, the estimator $\hat{\boldsymbol{\psi}}$ obeys classical likelihood theory: as $n\to\infty$, it is strongly consistent, asymptotically normal, attains the Cram\'er--Rao bound and converges at rate $\calO(n^{1/2})$ under well-known regularity conditions. Notice that with Model~2 in \eqref{ModelAD_AI}, the case $\beta=0$ is nonstandard as it lies on the boundary of the parameter space and must be treated separately; this issue, however, does not arise with Models~1~and~3 in~\eqref{ModelAD_AI3} and \eqref{ModelAD_AI2}, respectively. Moreover, the nonparametric transformation in \eqref{transformation} results in a slight asymptotically vanishing marginal misspecification for finite $n$. Despite this issue, \citet{Genest.etal:1995} show that under mild conditions, the maximum pseudo-likelihood estimator has similar asymptotic properties to the MLE, although with a slight loss in efficiency.

\subsection{Simulation study}\label{sec:SimulStudy}

To assess the performance of the maximum pseudo-likelihood estimator $\hat{\boldsymbol{\psi}}$ defined through \eqref{eq:lik}, we simulate $n=1000$ independent copies of the Gaussian scale mixture $X(\boldsymbol{s})=RW(\boldsymbol{s})$ defined in \eqref{RandomScale} at $D=5,10,15$ locations $\boldsymbol{s}_1,\ldots\boldsymbol{s}_D$ uniformly generated in $\calS=[0,1]^2$. Here, we focus on our random scale model \eqref{ModelAD_AI} (Model~2), as this is the model we choose to work with in our application in Section~\ref{sec:Application} for the reasons explained below. For the Gaussian process $W(\boldsymbol{s})$, we choose a stationary isotropic correlation function $\rho(\boldsymbol{s}_1,\boldsymbol{s}_2)=\exp\{-(\|\boldsymbol{s}_1-\boldsymbol{s}_2\|/\lambda)^\nu\}$ with range parameter ${\lambda=0.5,1}$ and smoothness parameter $\nu=1$. Sample paths of $X(\boldsymbol{s})$ are thus continuous but nondifferentiable. For the random scale variable $R$, we consider three scenarios: setting $\gamma=1$, we choose $\beta=0$ (asymptotic dependence) or $\beta=0.5,1$ (asymptotic independence). We then estimate all copula parameters $\boldsymbol{\psi}=(\lambda,\nu,\beta,\gamma)^T\in \Psi=(0,\infty)\times(0,2]\times[0,\infty)\times(0,\infty)$ by maximizing the censored pseudo-likelihood \eqref{eq:lik} using marginal thresholds $\boldsymbol{v}=(v,\ldots,v)^T$ with $v=0.95$. This yields $n(1-v)=50$ exceedances at each location, occurring simultaneously or not depending on the extremal dependence strength. We repeat this experiment $500$ times to obtain boxplots of estimated parameters and compute simple performance metrics. Computational details are described in the Supplementary Material.

\begin{table}[t!]
\centering
\begin{tabular}{c|cc|cc|cc}
&\multicolumn{2}{c|}{$100\times{\rm B}$} & \multicolumn{2}{|c|}{$100\times{\rm SD}$} & \multicolumn{2}{|c}{$100\times{\rm RMSE}$}\\
 & $\lambda=0.5$ & $\lambda=1$ & $\lambda=0.5$ & $\lambda=1$ & $\lambda=0.5$ & $\lambda=1$\\\hline
$\beta=0$ & $4/5/8/11$ & $12/7/11/19$ & $9/7/8/23$ & $19/6/10/27$ & $10/9/12/25$ & $22/9/15/33$ \\
$\beta=0.5$ & $4/2/8/25$ & $15/4/3/24$ & $10/7/23/51$ & $25/5/31/73$ & $11/7/24/57$ & $29/7/31/77$ \\
$\beta=1$ & $5/2/27/64$ & $14/3/18/58$ & $10/6/44/97$ & $23/5/53/108$ & $11/7/51/116$ & $27/6/56/123$ \\
\end{tabular}
\caption{Bias (${\rm B}$), standard deviation (${\rm SD}$) and root mean squared error (${\rm RMSE}$) of estimated parameters $\hat{\boldsymbol{\psi}}=(\hat\lambda,\hat\nu,\hat\beta,\hat\gamma)^T$ for the Gaussian scale mixture model \eqref{RandomScale} using \eqref{ModelAD_AI} with correlation function $\rho(\boldsymbol{s}_1,\boldsymbol{s}_2)=\exp\{-(\|\boldsymbol{s}_1-\boldsymbol{s}_2\|/\lambda)^\nu\}$ and parameters $\beta=0,0.5,1$, $\gamma=1$, $\lambda=0.5,1$ and $\nu=1$. Simulations are based on $n=1000$ independent replicates observed at $D=15$ uniform locations in $[0,1]^2$. Maximum likelihood estimation is based on \eqref{eq:lik} with threshold $\boldsymbol{v}=(v,\ldots,v)^T$, $v=0.95$. Results stem from $500$ independent experiments.}\label{tab:results1}
\end{table}

\begin{figure}[t!]
\centering
\includegraphics[width=\linewidth]{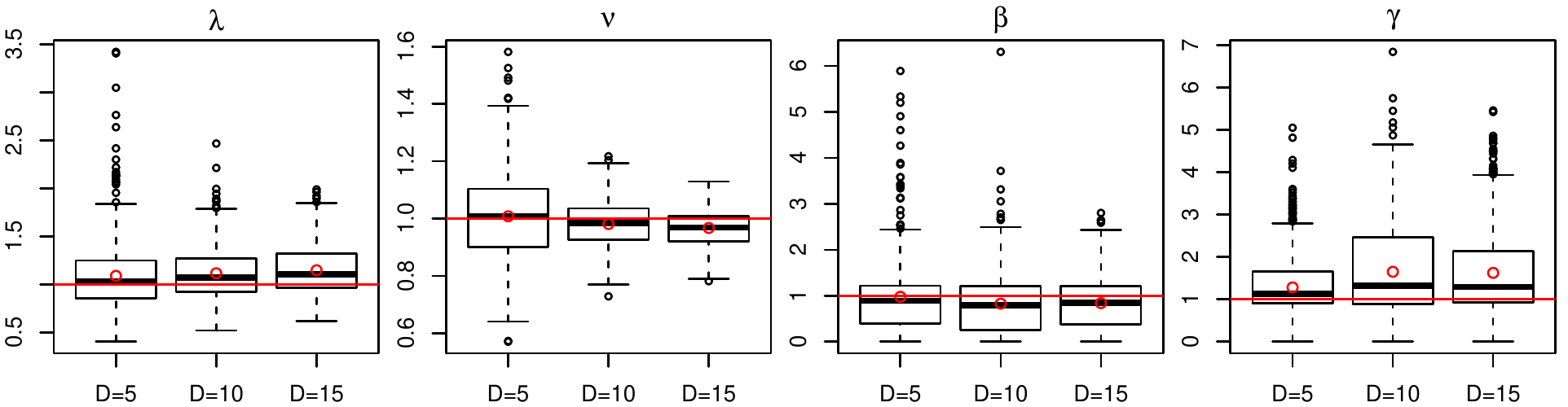}
\caption{Boxplots of estimated parameters $\hat{\boldsymbol{\psi}}=(\hat\lambda,\hat\nu,\hat\beta,\hat\gamma)^T$ for the Gaussian scale mixture model \eqref{RandomScale} using \eqref{ModelAD_AI} with correlation function $\rho(\boldsymbol{s}_1,\boldsymbol{s}_2)=\exp\{-(\|\boldsymbol{s}_1-\boldsymbol{s}_2\|/\lambda)^\nu\}$ and parameters $\lambda=\nu=\beta=\gamma=1$. Simulations are based on $n=1000$ independent replicates observed at $D=5,10,15$ uniform locations in $[0,1]^2$. Maximum likelihood estimation is based on \eqref{eq:lik} with threshold $\boldsymbol{v}=(v,\ldots,v)^T$, $v=0.95$. Boxplots are produced from $500$ independent experiments. Red points/horizontal lines show the estimated means/true values.}\label{fig:results1}
\end{figure}

Table~\ref{tab:results1} reports the absolute bias (${\rm B}$), standard deviation (${\rm SD}$) and root mean squared error (${\rm RMSE}=({\rm B}^2+{\rm SD}^2)^{1/2}$) of estimated parameters for $D=15$ based on the $500$ experiments. Figure~\ref{fig:results1} displays boxplots of estimated parameters for $\lambda=\nu=\beta=\gamma=1$ as a function of the dimension $D=5,10,15$. Generally, the estimation procedure appears to work well although the variability of estimates is relatively large owing to the fairly low number of exceedances considered in this simulation setting. The bias is almost always dominated by the standard deviation but not by much. The root mean squared error indicates that parameters seem overall easier to estimate for lower values of $\beta$ and $\lambda$. Indeed, when $\beta$ increases, the parameters $\beta$ and $\gamma$ are more complicated to identify because the resulting copula converges to the Gaussian copula as $\beta\to\infty$ or $\gamma\to\infty$, resulting in higher standard deviations and biases. Overall, standard deviations seem to decrease slightly as the dimension $D$ increases, even though $\boldsymbol{\psi}$ is not a consistent estimator for model \eqref{RandomScale} under infill asymptotics. The improvement is most striking for the smoothness parameter $\nu$, which is difficult to estimate with few scattered locations but is not clear for $\gamma$. The slightly increasing bias for larger $D$ is due to the nonparametric marginal estimation in \eqref{transformation}  as $n$ is kept fixed here, but it vanishes as $n\to\infty$. Figure~\ref{fig:results2} displays estimated $\chi(u)$ and $\bar\chi(u)$ coefficients in \eqref{eq:chi_chibar} when $\lambda=\nu=\beta=\gamma=1$ based either on simple nonparametric estimators or on our parametric censored likelihood approach with $D=5,10,15$ locations. From these plots, it is clear that the extremal dependence structure is well estimated, and that the rate of tail decay is relatively well captured. In contrast with nonparametric estimators, the model-based estimates have largely reduced uncertainties especially for large thresholds $v$ and borrow strength across locations for better tail estimation; hence they provide higher confidence about the asymptotic dependence class. Similar results were obtained with other parameter combinations under asymptotic independence ($\beta>0$) or dependence ($\beta=0$).

To investigate the tail flexibility of our model, we consider a misspecified setting. We simulate Student-$t$ random processes with ${\rm Df}=1,2,5,10$ degrees of freedom using the same correlation function, and we fit the Gaussian scale mixture \eqref{RandomScale} based on our model \eqref{ModelAD_AI}. Both models are Gaussian scale mixtures although they are not nested in each other. However, the joint tail behavior of a Student-$t$ process with ${\rm Df}$ degrees of freedom can be considered to be close to that of the Gaussian scale mixture model \eqref{ModelAD_AI} with $\beta=0$ and $\gamma={\rm Df}$ because the two models have the same limiting max-stable dependence structure. Our model should therefore provide a reasonable approximation although asymptotic dependence is a boundary case. Boxplots of estimated parameters (see Supplementary Material) suggest that when ${\rm Df}=1,2$, the performance of the MLE $\hat{\boldsymbol{\psi}}$ is similar to that obtained in a well-specified setting. When ${\rm Df}=5,10$, estimates of $\lambda$ and $\gamma$ are more ``biased,'' but the extremal dependence structure is still very well captured, as illustrated by estimated $\chi(u)$ and $\bar\chi(u)$ coefficients in Figure~\ref{fig:results3}.
\begin{figure}[t!]
\centering
\includegraphics[width=0.9\linewidth]{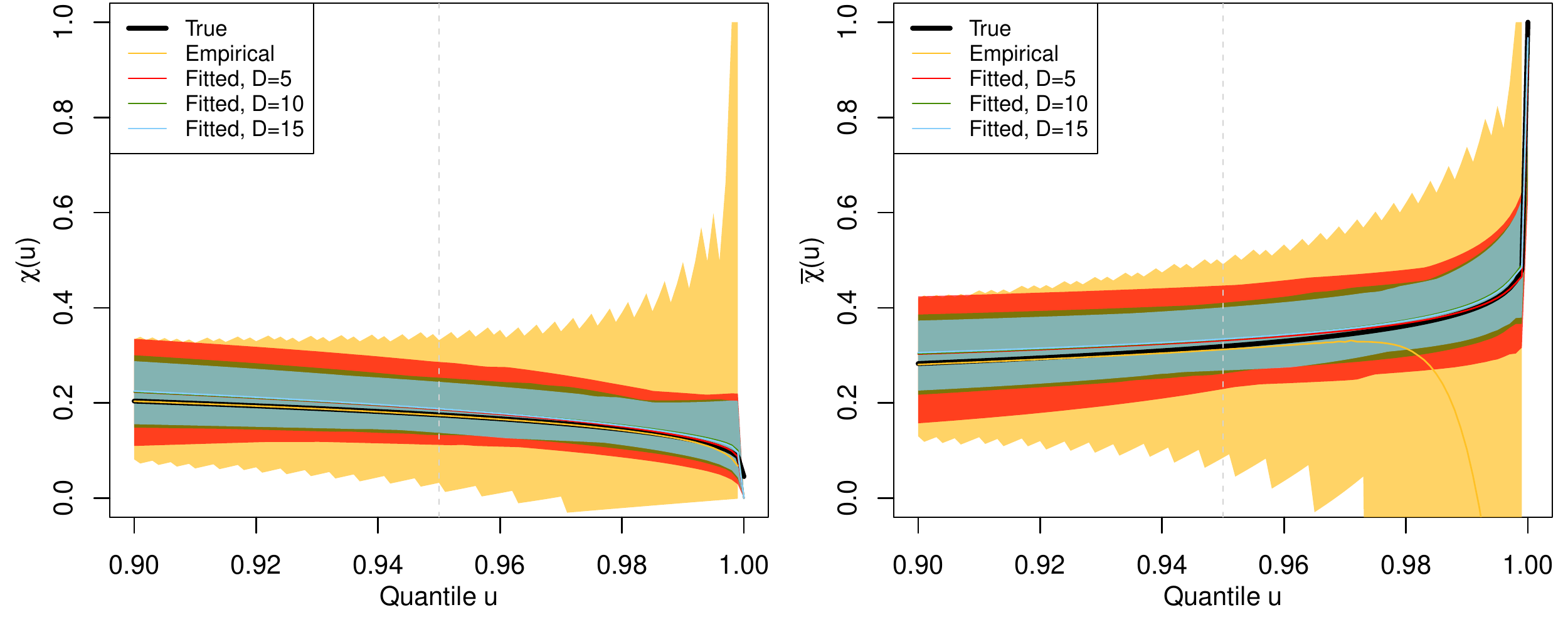}
\caption{Estimated coefficients $\chi(u)$ (left) and $\bar\chi(u)$ (right), $u\in[0.9,1]$, for data simulated from a Student-$t$ process with ${\rm Df}=10$ degrees of freedom and correlation function $\rho(\boldsymbol{s}_1,\boldsymbol{s}_2)=\exp\{-(\|\boldsymbol{s}_2-\boldsymbol{s}_1\|/\lambda)^\nu\}$ with $\lambda=0.5$ and $\nu=1$ for two points at distance $\|\boldsymbol{s}_2-\boldsymbol{s}_1\|=0.5$. Estimation is either nonparametric (yellow) or parametric based on the (misspecified) Gaussian scale mixture model \eqref{RandomScale} with random scale \eqref{ModelAD_AI}, using $D=5$ (red), $10$ (green) or $15$ (blue) uniform locations in $[0,1]^2$. The number of replicates is $n=1000$. Solid lines show means of $500$ simulations, while shaded areas are $95\%$ overall confidence envelopes. True curves are in solid black, and the threshold $v=0.95$ used in \eqref{eq:lik} is the vertical dashed line.}\label{fig:results3}
\end{figure}
In summary, our Gaussian scale mixture model provides a very flexible tail-dependent structure described by a relatively small number of parameters and can be consistently estimated from high threshold exceedances using a full pseudo-likelihood approach with partial censoring but remains fairly computer intensive for large dimensions.

\section{Application}\label{sec:Application}
To illustrate the benefits of our new modeling approach, we analyze hourly wind speed extremes recorded during 2012--2014 in the Pacific Northwest, US, a region with large wind energy resources. 
The dataset was compiled from data archived by the Bonneville Power Administration. Data are available year-round at $20$ meteorological towers along the border between Oregon and Washington. To avoid modeling complex spatio-temporal non-stationary patterns, we restrict our attention to winter months (DJF) for the 12 stations located on the East side of the Cascade mountain range; see Figure~\ref{fig:WindMap}. Selected data comprise up to $6504$ hourly observations at each site, with about $8\%$ of values missing. Owing to important East-West pressure gradients in this region and the special nature of the orography, wind patterns are mainly characterized by easterly and westerly winds (
\citealp{Kazor.Hering:2015a}). 
A wind rose for the data at the 12 stations reveals that extreme winds blow mostly from the West or South-West, suggesting that simple anisotropic models might perform well; see Figure~\ref{fig:WindMap}. More details about data and monitoring stations are in the Supplementary Material.

\begin{figure}[t!]
\begin{minipage}{0.63\linewidth}
\centering
\includegraphics[width=\linewidth]{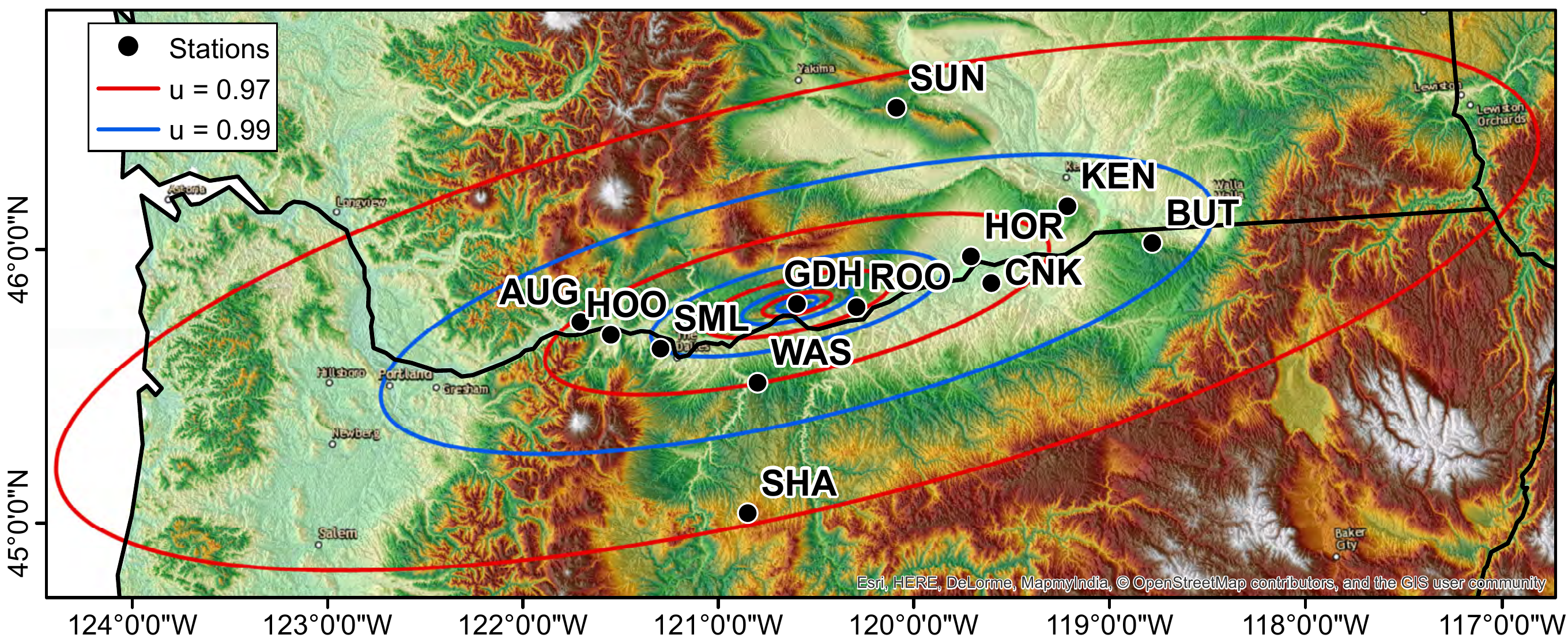}
\end{minipage}\hfill\begin{minipage}{0.37\linewidth}

\vspace{-5pt}

\includegraphics[width=\linewidth]{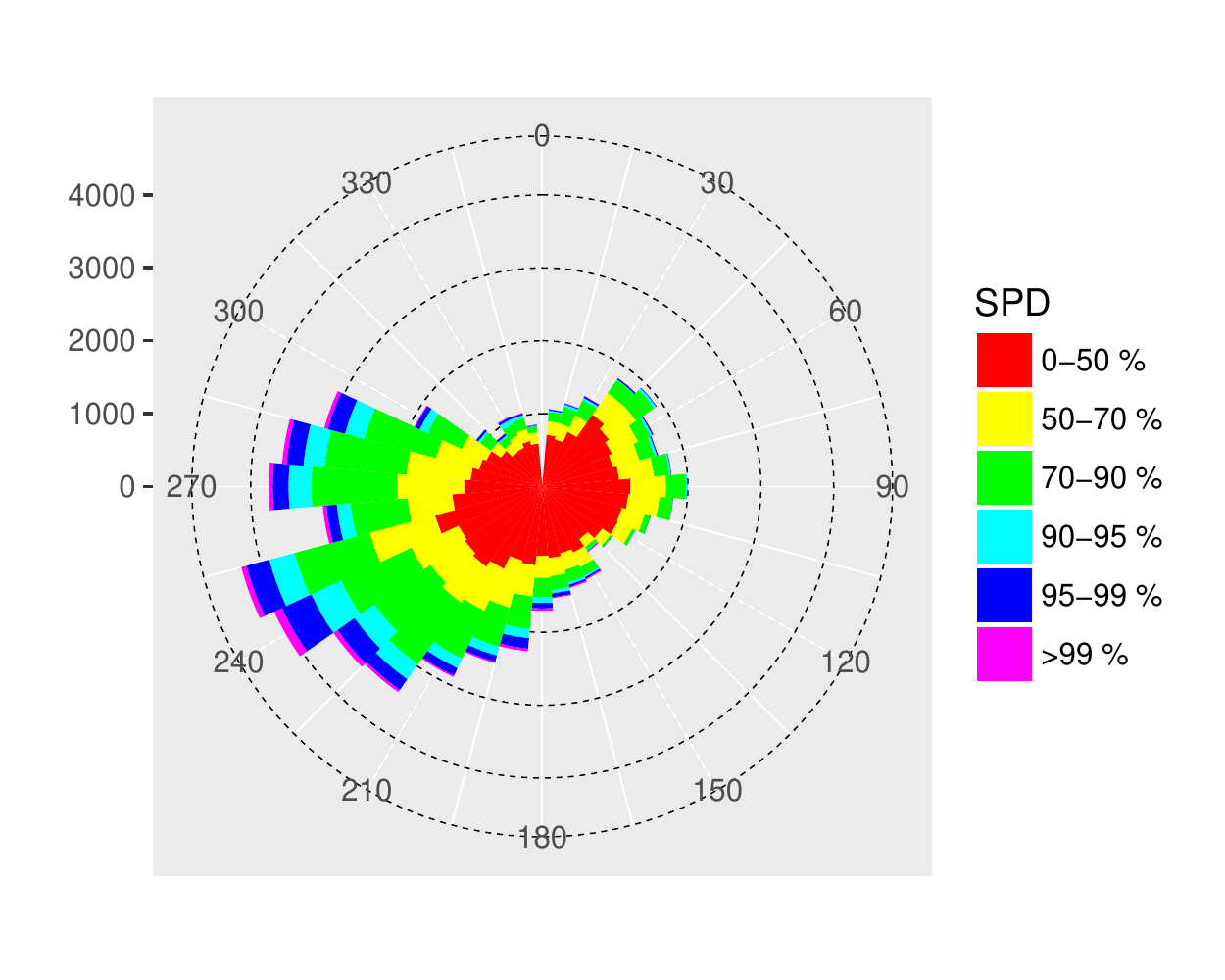}
\end{minipage}
\caption{\emph{Left:} Topographic map with meteorological towers selected in our study and state boundaries (black). The ellipses overlaid, centered at Goodnoe Hills (GDH), correspond to the isocontours of the fitted coefficients $\chi(u)=0.5,0.4,0.3,0.2$ (from the center outward) using the thresholds $u=0.97$ (red) and $u=0.99$ (blue), for the best anisotropic model. \emph{Right:} Wind rose of winter (DJF) wind speeds for the 12 stations, preliminarily transformed to the uniform scale. The color scale corresponds to different marginal quantile ranges.}\label{fig:WindMap}
\end{figure}

Hourly wind speeds show strong temporal dependence. Since our focus is on spatial dependence, we ignore temporal dependence for parameter estimation but account for it in uncertainty assessment using a block bootstrap. We consider three model classes:
\begin{enumerate}
\item[1.]  RW: Our new Gaussian scale mixture copula model defined through \eqref{RandomScale} and \eqref{ModelAD_AI} (Model~2), with underlying anisotropic correlation function $\rho(\boldsymbol{s}_1,\boldsymbol{s}_2)=\exp\left\{-(h_{12}/\lambda)^\nu\right\}$, where $h_{12}$ denotes the Mahalanobis distance, that is,
\begin{equation*}
h_{12}^2=(\boldsymbol{s}_1-\boldsymbol{s}_2)^T\Omega^{-1}(\boldsymbol{s}_1-\boldsymbol{s}_2),\quad \Omega=\begin{pmatrix}
\cos(\theta) & -\sin(\theta)\\
\sin(\theta) & \cos(\theta)
\end{pmatrix}\begin{pmatrix}
1 & 0\\
0 & \lambda_{12}^{-2}
\end{pmatrix}\begin{pmatrix}
\cos(\theta) & -\sin(\theta)\\
\sin(\theta) & \cos(\theta)\\
\end{pmatrix}^T.
\end{equation*}
The Mahalanobis distance defines elliptical isocontours; $\lambda$ and $\lambda_{12}$ are respectively the length of one principal axis and the length ratio of the two principal axes, while $\theta$ is the angle with respect to the West-East direction. When $\lambda_{12}=1$, the model is isotropic. The reason why we consider Model~2 instead of Models 1 or 3 proposed in \S\ref{InterpAIandAD} is that extreme wind speeds (and wind gusts in particular) are known to be very localized, suggesting that a very flexible asymptotically independent model might perform well in practice. However, notice that Model~2 does not exclude the possibility of asymptotic dependence, and that any uncertainty assessment based on Model~2 will take this possibility into account.
\item[2.]  $t$: The Student-$t$ copula with ${\rm Df}=1,\ldots,30$ degrees of freedom and the same anisotropic correlation function $\rho(\boldsymbol{s}_1,\boldsymbol{s}_2)$.
\item[3.]  Gauss: The Gaussian copula (same as 2., but with ${\rm Df}=\infty$ fixed).
\end{enumerate}
Eight models, summarized in the Supplementary Material, were fitted to the wind data using the censored likelihood estimator with threshold $\boldsymbol{v}=(v,\ldots,v)$, $v=0.95$; see \S\ref{subsec:Lik}. This yields about $6504\times 0.05\approx 325$ correlated exceedances at each site. The log profile censored likelihood of the (an)isotropic Student-$t$ copula is shown in Figure~\ref{fig:WindLikelihoods} and compared to the maximized log likelihoods of the (an)isotropic Gaussian scale mixture and Gaussian copula.
\begin{figure}[t!]
\centering
\begin{minipage}{0.5\linewidth}
\centering
\includegraphics[width=0.9\linewidth]{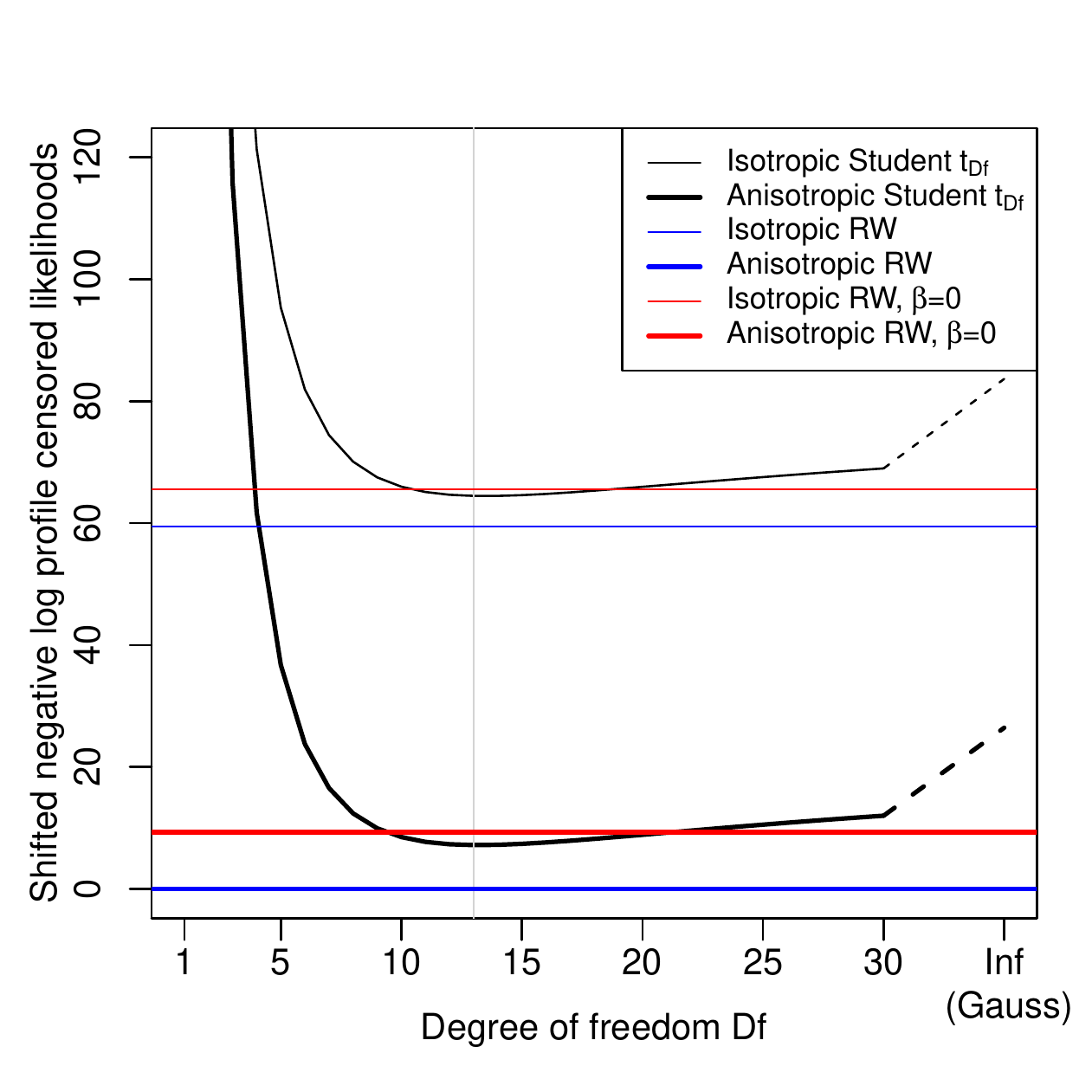}
\end{minipage}\hfill\begin{minipage}{0.5\linewidth}
\centering
\includegraphics[width=0.9\linewidth]{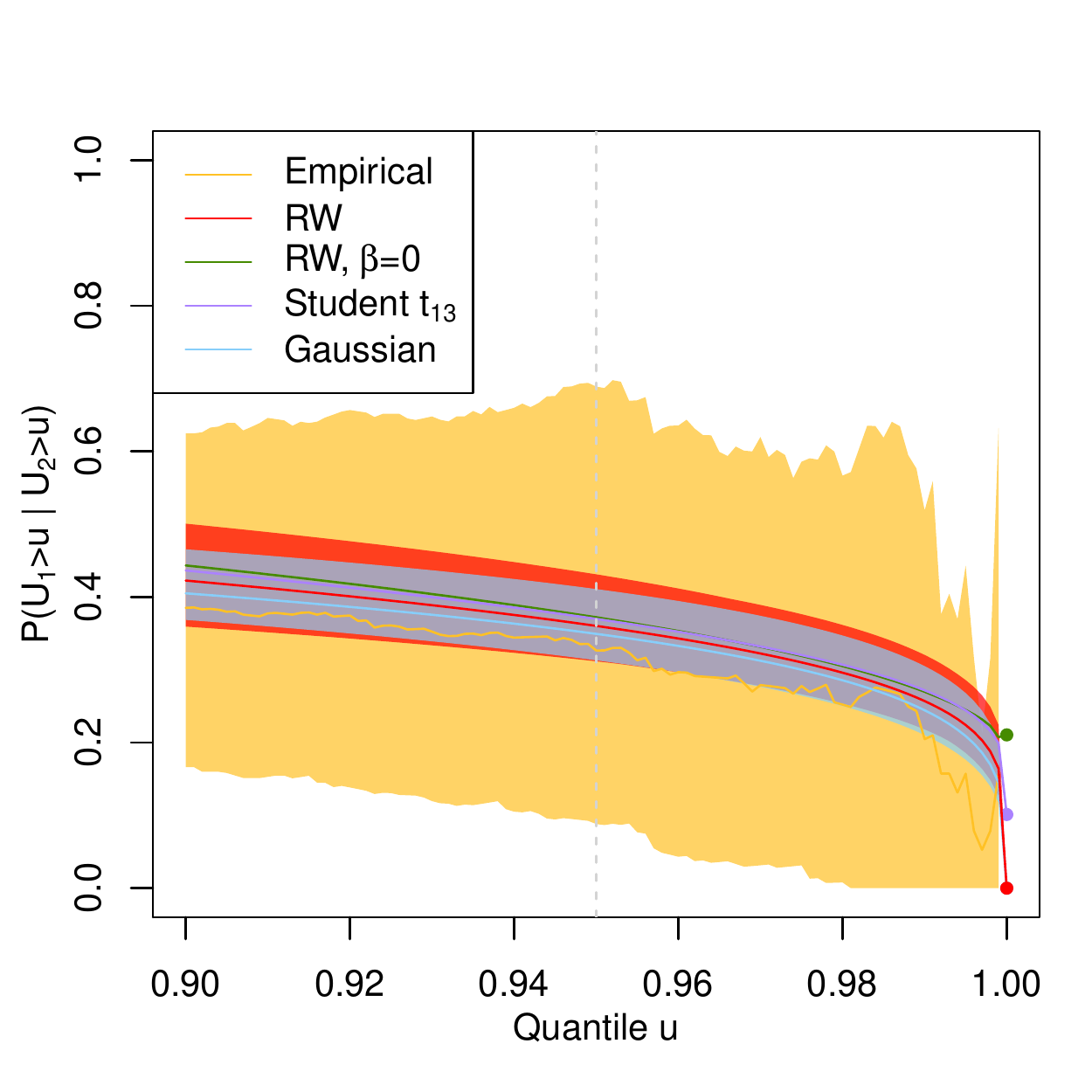}
\end{minipage}
\caption{\emph{Left:} Negative log profile censored likelihood for the Student-$t$ copula with ${\rm Df}$ degrees of freedom (black), and negative maximized log censored likelihood for the Gaussian scale mixture copula with $\beta\geq0$ (blue) and $\beta=0$ (red). The Gaussian copula corresponds to ${\rm Df}=\infty$. Isotropic and anisotropic models correspond to thin and thick lines, respectively. The best log likelihood value has been subtracted from all curves. The vertical grey line at ${\rm Df}=13$ represents the best Student-$t$ model. \emph{Right:} Estimated probability ${\pr(U_1>u\mid U_2>u)}$, as a function of the threshold $u$, with $U_1,U_2$ representing the hourly wind speed data at Goodnow Hills (GDH) and Augspurger (AUG), preliminarily transformed to the uniform scale. Different lines display the empirical estimate (yellow), and the fitted curves for the anisotropic unrestricted RW (red), restricted RW with $\beta=0$ (green), the Student-$t$ with ${\rm Df}=13$ (purple) and Gaussian (blue) models. $95\%$ overall confidence envelopes are shown for the empirical, unrestricted RW and Gaussian fits (using similar colors). Similar envelopes were obtained for the other models but are omitted here for readability.}\label{fig:WindLikelihoods}
\end{figure}
Deviances are not straightforward to interpret owing to temporal dependence, but anisotropic models seem to outperform isotropic models by a large margin. This is not surprising given the wind patterns in the study region. Moreover, the best Student-$t$ model has ${\rm Df}=13$ degrees of freedom, a result that is consistent for isotropic and anisotropic models. This suggests that the asymptotic dependence strength is quite weak although the asymptotically independent Gaussian copula model is outperformed by the extra flexibility of the Student-$t$ copula. The best model overall is the unrestricted anisotropic Gaussian scale mixture copula (with $\beta\geq0$); the difference in log likelihoods is about $7$ and $9$, respectively, with respect to the best Student-$t$ model and the unrestricted Gaussian scale mixture model with $\beta=0$. Estimated parameters for our best model, with $95\%$-confidence intervals computed using a weekly block bootstrap with $100$ replicates respecting the missing value patterns are reported in Table~\ref{tab:estim.par}. 
\begin{table}[t!]
\centering
\caption{Estimated parameters $\hat{\boldsymbol{\psi}}=(\hat\lambda,\hat\lambda_{12},\hat\theta,\hat\nu,\hat\beta,\hat\gamma)^T$ for the unrestricted anisotropic Gaussian scale mixture copula model with corresponding $95\%$-confidence intervals (CIs).}\label{tab:estim.par}
\begin{tabular}{r|cccccc}
Parameters & $\lambda$ [km] & $\lambda_{12}$ & $\theta$ [rad] & $\nu$ & $\beta$ & $\gamma$\\
\hline
Estimates & $287$ & $3.71$ & $0.23$ & $0.46$ & $1.96$ & $0.05$\\
$95\%$-CIs & $[95,1127]$ & $[2.64,5.08]$ & $[0.18,0.34]$ & $[0.32,0.76]$ & $[0.96,8.47]$ & $[0,0.32]$ \\
\end{tabular}
\end{table}
\begin{figure}[t!]
\centering
\includegraphics[width=\linewidth]{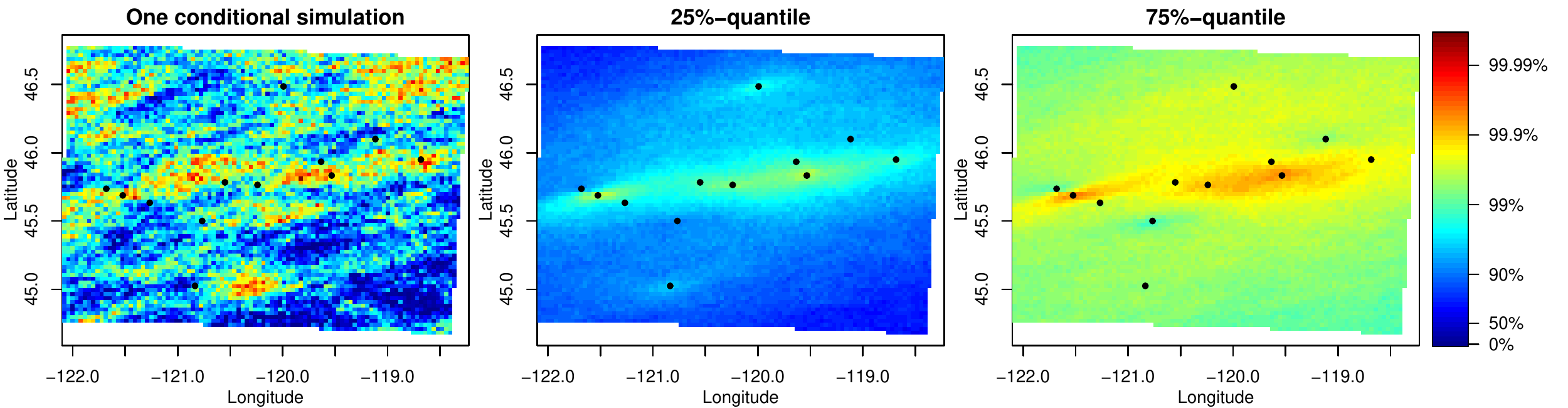}
\caption{\emph{Left:} Conditional simulation over the region of study for our best fitted RW copula model with parameters given in Table~\ref{tab:estim.par}. The simulation is done based on the algorithm described in \S\ref{sec:condsimalgo}, conditionally on the values observed at the twelve stations (black dots) on February 22, 2012, a day of very strong wind. \emph{Middle and right:} corresponding $25\%$ and $75\%$-conditional quantiles based on $500$ conditional simulations. The color scale indicates quantile probabilities.}
\label{fig:WindCondSim}
\end{figure}
Figure~\ref{fig:WindCondSim} shows a conditional simulation and corresponding $25\%$ and $75\%$-conditional quantiles (calculated based on $500$ simulations using the algorithm described in \S\ref{sec:condsimalgo}) of the wind speed field over the region of study for February 22, 2012, which corresponds to the day of strongest average wind speeds observed at the twelve stations.

The estimate $\hat\lambda_{12}=3.71$ implies that the anisotropy is strong, and $\hat\theta=0.23$ indicates that the direction of strongest correlation is toward the North-East, as illustrated by the isocontours of the fitted coefficient $\chi(u)$ overlaid on the left panel of Figure~\ref{fig:WindMap}, and by the wind patterns simulated in Figure~\ref{fig:WindCondSim}. Interestingly, this agrees with the wind rose, although the wind directions were not used in the fitting procedure. 
The value $\hat\nu=0.46$ indicates that the hourly wind field shows small scale variability, as expected. The estimates $\hat\beta=1.96$ and $\hat\gamma=0.05$ for the random scale parameters strongly support the assumption of asymptotic independence, given that the confidence interval for $\beta$ excludes zero by far. Furthermore, for all bootstrap replicates, the likelihood value was always lower when fixing $\beta=0$. This also indicates that standard asymptotics (with $\beta>0$) prevail in this case. The values of $\hat\beta$ and $\hat\gamma$ also suggest that the extremal dependence structure is relatively far from being Gaussian, and that our model better captures the data's extremal properties. To further validate our fitted models, the right panel of Figure~\ref{fig:WindLikelihoods} shows the estimated probability $\pr(U_1>u\mid U_2>u)$, which is more easily interpretable than $\chi(u)$ yet asymptotically equivalent as $u\to1$, plotted as a function of the threshold $u\in[0.9,1]$, where $U_1$ and $U_2$ represent the hourly wind speed data at Goodnow Hills (GDH) and Augspurger (AUG), preliminarily transformed to the uniform scale. This plot suggests that although all models seem to perform similarly and reasonably well at moderately high levels, they nevertheless predict very different behaviors for very large extremes with $u\to1$, highlighting the need for flexible tail models covering asymptotic independence and dependence scenarios, such as those proposed in \S\ref{InterpAIandAD}. This figure also clearly demonstrates that our model-based approach allows for a considerable reduction in the uncertainty compared to a fully nonparametric approach. 

\section{Discussion}\label{sec:Discussion}
Starting with a detailed study of  the tail behavior of general Gaussian scale mixtures, we have proposed new parsimonious and flexible subasymptotic copula models to achieve  a smooth transition between asymptotic independence and asymptotic dependence. Unlike approaches often used in multivariate analysis confronting the asymptotically dependent Student-$t$ model to the asymptotically independent Gaussian model, we give strong attention to appropriately capturing the tail decay in asymptotically independent scenarios while keeping a highly flexible asymptotically dependent submodel. Although our main contributions in this paper concern tail characteristics of Gaussian scale mixtures and their estimation, these elliptic copula models may also be useful for the modeling of the full data range, arising in a much wider spectrum of applications. 

In addition to providing more flexibility than max-stable models, inference for our model is also facilitated and may be performed using a censored pseudo-likelihood, although the latter is expressed in terms of integrals whose numerical approximation is quite intensive to compute. 
Overall, computations are of the same order as for the censored Poisson and Pareto likelihoods advocated by \citet{Wadsworth.Tawn:2014} and \citet{Thibaud.Opitz:2015}, which are only valid for asymptotically dependent data. The asymptotically dependent submodels of the models that we propose are closely related to the elliptic Pareto process of \citet{Thibaud.Opitz:2015} while their asymptotically independent counterparts provide alternatives and more flexible parametric extensions (Models 2 and 3) to the Laplace model of \citet{op2015}. Estimating the value of the random scale shape parameter allows the data to provide evidence about the asymptotic dependence class without fixing it a priori. 

Our partially censored likelihood approach transcribes the common idea in statistics of extremes that only observations in the tail should be used to determine the extremal dependence structure. It would of course be possible to use milder assumptions that yield faster estimation procedures, for example by using the full uncensored density when the observation vector is a partial exceedance of the multivariate threshold or by using robust  M-estimators based on ranks and bivariate tail characteristics as proposed in \citet{Einmahl.al:2016} for max-stable models  to avoid costly high-dimensional numerical integration.

Our application demonstrates that our new model can be useful in practice. Analyzing hourly wind speed data in the Pacific Northwest, US, we found quite strong evidence of asymptotic independence, and we showed that the extremal dependence structure was well captured by our copula model \eqref{ModelAD_AI}. Over large regions, a limitation of our model may be that independence at large distances can only be captured through exact Gaussian submodels, as a common random scale induces dependence to independent Gaussian components. To circumvent this issue, our model could be extended by considering a random partitioning approach \citep{Morris.etal:2017} or a random set element \citep{Huser.Davison:2014a} but inference would become awkward. Modeling of non-stationarity \citep{Huser.Genton:2016} and space-time dependence \citep{Huser.Davison:2014a} is an important aspect in statistics of extremes; it would be interesting to investigate useful extensions of Gaussian scale mixture models in further research.

\section*{Acknowledgments}
We thank Amanda Hering (Baylor University) for sharing the wind data and Luigi Lombardo (KAUST) for cartographic support. This work was undertaken while Emeric Thibaud was at Colorado State University with partial support by US National Science Foundation Grant DMS-1243102. Thomas Opitz was partially supported by the LEFE-CERISE project coordinated by the French National Center for Scientific Research.
\section*{Appendix: Proofs}
\appendix
We first provide lemmas from the literature by unifying notations and adapting existing results to our context.

\begin{lemma}[Gumbel domain of attraction; \citealp{ha2006b,ha2010a}]\label{lem:gumbel}\ 
	For a random variable $X\sim F_X$ with infinite upper endpoint, suppose that a weight function ${w(\cdot)>0}$ satisfies
	\begin{equation*}
	{1-F_X\{u  + x/w(u)\}\over 1-F_X(u)} \rightarrow \exp(-x), \quad u\rightarrow\infty.
	\end{equation*}
	Then $F_X$ is in the maximum domain of attraction of the unit Gumbel distribution, that is
	\begin{equation}\label{eq:convgumb}
	F_X^n(a_nx+b_n)\rightarrow \exp\{-\exp(-x)\},\qquad x\in\Real, 
	\end{equation}
	where we can choose  $b_n=F^{-1}(1-1/n)$ and $a_n=1/w(b_n)$ for $n>1$.
	If $X$ has Weibull-type tail decay 
	\begin{equation*}
	\pr(X\geq x) \sim \alpha x^{\gamma} \exp(-\delta x^\beta), \qquad x\to\infty,
	\end{equation*}
	for some constants $\alpha>0$, $\beta>0$, $\gamma\in \Real$ and   $\delta>0$, then \eqref{eq:convgumb} holds and we may take
	\begin{equation}\label{eq:wgumb}
	w(u) \sim \delta\beta u^{\beta-1},\quad u\rightarrow \infty.
	\end{equation}
\end{lemma}

\begin{lemma}[Joint tail decay for Weibull-type random scale; \citealp{ha2010a}]\label{lem:weibjoint}\ 
	Assume that the  bivariate elliptic random vector $\boldsymbol{X}=(X_1,X_2)^T\Deq R\boldsymbol{\Sigma}^{1/2}(U_1,U_2)^T$ with correlation matrix $\boldsymbol{\Sigma}$, $\rho=\boldsymbol{\Sigma}_{1;2}\in(-1,1)$, has radial distribution $R\sim F$ in the Gumbel domain of attraction with weight function $w(\cdot)$ as in \eqref{eq:wgumb} and with copula $C$. Then  
	\begin{equation}\label{eq:CRV}
	\dfrac{\overline{C}(1-x_1/u,1-x_2/u)}{\overline{C}(1-1/u,1-1/u)}\rightarrow (x_1x_2)^{1/(2\eta)}, \qquad \eta=\{(1+\rho)/2\}^{\beta/2}, \quad u\rightarrow \infty,
	\end{equation}
	where $\eta$ is the coefficient of tail dependence of $\boldsymbol{X}$. 
	More specifically, if $R$ has Weibull-type tail with parameters as in \eqref{eq:AIV}, then
	\begin{align}
	\label{eq:jointC}
	&\overline{C}(1-1/x,1-1/x) \sim K_1 \left(\log x\right)^{K_2}x^{-1/\eta}, \qquad x\rightarrow\infty, \\
	&K_1=\alpha^{-1} (1-\rho)^{-2}\{2/(1+\rho)\}^{\gamma/2-\beta/2+1} (1-\rho^2)^{3/2}\delta^{(1/\eta-1)\gamma/\beta} \{\alpha^2/(2\pi\beta)\}^{1-1/(2\eta)},\nonumber\\ 
	&K_2= (1-1/\eta)\gamma/\beta+1/(2\eta)-1. \nonumber 
	\end{align}
\end{lemma}

\begin{lemma}[Tail decay of products of Weibull-type variables; \citealp{ar2011}]\label{lem:weibprod}\ 
	If two independent random variables $R_1$ and $R_2$ have Weibull-type tails such that 
	\begin{equation*}
	\pr(R_i\geq r) \sim \alpha_{i} r^{\gamma_i} \exp(-\delta_{i} r^{\beta_i}), \quad r\to\infty, \quad i=1,2,
	\end{equation*}
	where $\alpha_{i}>0$, $\beta_i>0$, $\gamma_i\in\mathbb{R}$, $\delta_i>0$ for $i=1,2$,
	then, as $r\rightarrow\infty$, 
	\begin{equation}\label{eq:weibprod}
	\pr(R_1R_2\geq r) \sim \sqrt{2\pi} \sqrt{\dfrac{\beta_2\delta_{2}}{\beta_1+\beta_2}}\alpha_{1}\alpha_{2}A^{0.5\beta_2+\gamma_2-\gamma_1}r^{\frac{2\beta_2\gamma_1+2\beta_1\gamma_2+\beta_1\beta_2}{2(\beta_1+\beta_2)}}\exp\left(-B r^{\frac{\beta_1\beta_2}{\beta_1+\beta_2}} \right),
	\end{equation}
	where $A=\{(\beta_1\delta_{1})/(\beta_2\delta_{2})\}^{1/(\beta_1+\beta_2)}$, 
	$B=\delta_{1}^{1-a}\delta_{2}^a\{(\beta_2/\beta_1)^a+(\beta_1/\beta_2)^{1-a}\}$ and $a=\beta_1/(\beta_1+\beta_2)$.
\end{lemma}

\begin{proof}
	Without loss of generality, consider $i=1$. We have $X_1\overset{D}{=} \sqrt{\sigma_{11}} R U_1$ with $U_1^2\sim \mathrm{Beta}\{1/2,(D-1)/2\}$, see \citet[Lemma 2]{ca1981}. The tail representation \eqref{eq:AIV} is stable under power transformations of variables like $R\mapsto R^{p}$, $p\in(0,\infty)$, where some of the parameters in  \eqref{eq:AIV} are different for the transformed variable $R^p$. We can fix $p=2$ and  apply Theorem 4.1 of \citet{ha2010b} to $X_1^2$ and $R^2$, which provides results concerning the random scaling with a $\mathrm{Beta}$-distributed variable. Indeed, $X_1^2$ is Weibull-tailed if $R^2$ is so, and  $X_1^2$ and $R^2$ have the same weight function $w(\cdot)$ in \eqref{eq:wgumb} and therefore the same parameter $\beta$. Switching  back to $|X_1|$ and $R$ using $p=1/2$ and taking into account the symmetry of $X_1$ with respect to $0$ then proves the statement of the lemma.
\end{proof}

\begin{proof}[Proof of Theorem~\ref{theor:cond}] 
	The formulas for conditional elliptic distributions follow from the related standard theory 
	\citep[][Corollary 5 and Section 4]{ca1981}. The expression for the density $f_{RR_{\boldsymbol{W}}}$ of the radial variable of a Gaussian scale mixture distribution is obtained as follows:
	\begin{align*}
	f_{RR_{\boldsymbol{W}}}(r) &= \dfrac{\partial}{\partial r} \int_0^\infty F(r/s)f_{R_{\boldsymbol{W}}}(s)\,\mathrm{d}s \;=\;\int_0^\infty F(1/s) \dfrac{\partial}{\partial r} f_{R_{\boldsymbol{W}}}(rs)r\,\mathrm{d}s \\
	&=\int_0^\infty  F(1/s)\left\{ f_{R_{\boldsymbol{W}}}(rs) +f'_{R_{\boldsymbol{W}}}(rs)rs\right\}\,\mathrm{d}s.
	\end{align*}
	The density  of $R$ conditional to  $\boldsymbol{X}=\boldsymbol{x}$ is based on the change of variables from $(r, \boldsymbol{w})$ to $(r, \boldsymbol{x})=(r,r \boldsymbol{w})$ in
	$f_{R, \boldsymbol{W}}=f\times \varphi_D(\boldsymbol{\cdot};\boldsymbol{\Sigma})$:
	\begin{equation*}
	f_{R, \boldsymbol{X}}(r,\boldsymbol{x})=f(r)\times \varphi_D(\boldsymbol{x}/r;\boldsymbol{\Sigma}) \times r^{-D}.
	\end{equation*}
	The formula $f_{R\mid  \boldsymbol{X}= \boldsymbol{x}}(r)=f_{R, \boldsymbol{X}}(r, \boldsymbol{x})/f_{\boldsymbol{X}}(\boldsymbol{x})$ yields the conditional density.\\ 
\end{proof}

\begin{proof}[Proof of Theorem~\ref{theor:tailindep}] \ \\
	The bivariate standard Gaussian vector $\boldsymbol{W}=(W_1,W_2)^T$ has elliptic representation $R_W (U_1, \rho U_1+\sqrt{1-\rho^2}U_2)^T$ with $R_W\indep (U_1,U_2)^T$ and $(U_1,U_2)^T$ distributed uniformly on the unit circle. In the following, parameter and variable subscripts, such as $R_W$ and $\alpha_W$, always refer to the vector $\boldsymbol{W}$. Since $R_W^2 \sim \chi^2_2 = \Gamma_{\beta=1}(a=1,b=2)$ is gamma distributed,  we transform the gamma tail decay rate from Lemma 1 in the Supplementary Materials to obtain
	\begin{equation*}
	\pr(R_W> r) = \pr(R_W^2 > r^2)  \sim (r/2)^{1-1}\{\Gamma(1)\}^{-1}\exp(-r^2/2)=   \exp(-r^2/2), \quad r\rightarrow \infty.
	\end{equation*}
	Therefore, $R_W$ has  Weibull-type tail as in \eqref{eq:AIV} with parameters $\alpha_{W}=1, \  \beta_{W}=2, \gamma_W=0,\ \delta_{W}=1/2$ in obvious notation. 
	We write $\boldsymbol{X}=R\boldsymbol{W} = R^{\star} (U_1, \rho U_1+\sqrt{1-\rho^2}U_2)^T$ with $R^{\star}=RR_W$ a product of two independent Weibull-type variables. 
	The tail expansion  for such products, given in Lemma~\ref{lem:weibprod},  yields a Weibull-type tail as in \eqref{eq:AIV} for $R^\star$. 
	Setting $R_1=R$ and $R_2=R_W$ in Lemma~\ref{lem:weibprod}, the constant in \eqref{eq:weibprod} is $A = \left\{ (\beta \delta)/(2\times 1/2)\right\}^{1/(\beta+2)} =  (\beta \delta)^{1/(2+\beta)}$ and
	\begin{equation*}\label{q:Rstartail}
	\pr(R^\star > r) \sim \alpha^\star r^{\gamma^\star} \exp(-\delta^\star r^{\beta^\star}),\quad r\to\infty,
	\end{equation*}
	with 
	\begin{align}
	\alpha^\star &=  \left(\dfrac{2\pi \beta_W \delta_W}{\beta_W+\beta}\right)^{1/2} A^{\beta_W/2+\gamma_W-\gamma} \alpha \alpha_{W} = \left(\dfrac{2\pi}{2+\beta}\right)^{1/2}  A^{1-\gamma} \alpha,\nonumber\\
	\beta^\star &= 2\beta/(2+\beta), \nonumber\\
	\gamma^\star &=  \dfrac{2\gamma+\beta}{2+\beta}, \nonumber\\
	\delta^\star &= \delta^{1-\beta/(\beta_W+\beta)}\delta_W^{\beta/(\beta_W+\beta)}\left\{ (\beta_W/\beta)^{\beta/(\beta_W+\beta)}+ (\beta/\beta_W)^{1-\beta/(\beta_W+\beta)} \right\}  \nonumber\\
	&=\delta^{2/(2+\beta)}2^{-\beta/(2+\beta)} \left\{ (2/\beta)^{\beta/(2+\beta)}+ (\beta/2)^{2/(2+\beta)} \right\}. \nonumber 
	\end{align}
	We derive $\bar{\chi}$ using Lemma \ref{lem:weibjoint}.  
	The weight function of $R^\star$ in \eqref{eq:wgumb}  is $w^\star(u)=\delta^\star\beta^\star u^{\beta^\star-1}$.  Lemma~\ref{lem:weibjoint} yields the coefficient of tail dependence of $\boldsymbol{X}$ given as $\eta=\{(1+\rho)/2\}^{\beta/(2+\beta)}$. We obtain $\bar{\chi} = 2(\eta-1)= 2\{(1+\rho)/2\}^{\beta/(\beta+2)}-1$.
	
	Lemma \ref{lem:weibjoint} helps to derive the joint tail representation \eqref{eq:AIbv}. 
	By combining  \eqref{eq:CRV} and \eqref{eq:jointC} (with $x_1=x_2=1$) and substituting $R$ in Lemma \ref{lem:weibjoint} by $R^\star$ (with its tail parameters $\alpha^\star,\beta^\star,\gamma^\star,\delta^\star$), we obtain that $\overline{C}(1-1/x,1-1/x) \sim\overline{C}(1-1/x,1-1/x)=K \left(\log x\right)^{ (1-1/\eta)\gamma^\star/\beta^\star+1/(2\eta)-1} x^{-1/\eta}$ as $x\rightarrow\infty$, with  constant term
	\begin{equation*}\label{eq:jointK}
	K=(\alpha^\star)^{-1} (1-\rho)^{-2} \{2/(1+\rho)\}^{\gamma^\star/2-\beta^\star/2+1} (1-\rho^2)^{3/2}(\delta^\star)^{(1/\eta-1)\gamma^\star/\beta^\star} \left\{(\alpha^\star)^2/(2\pi\beta^\star)\right\}^{1-1/(2\eta)},
	\end{equation*}
	which proves \eqref{eq:AIbv}.
	Finally, $\chi=0$ follows from $\bar{\chi} < 1$. 
\end{proof}

The proof of Theorem \ref{theor:taildep} is obtained using Breiman's lemma \citep{br1965}; see \citet{op2013}.



\baselineskip=14pt

\bibliographystyle{CUP}
\bibliography{Biblio}

\baselineskip 10pt

\end{document}